\documentclass[twocolumn]{autart}    
\usepackage{cite}
\usepackage{amsmath,amssymb,amsfonts,mathtools}
\usepackage{algorithm,algorithmic}

\usepackage{graphicx,subcaption}

\usepackage{tikz,tikz-3dplot}
\usetikzlibrary{arrows,arrows.meta,shapes,chains,calc,positioning}
\usepackage{threeparttable, makecell, booktabs, multirow, array}

\newtheorem{definition}{Definition}
\newtheorem{problem}{Problem}
\newtheorem{assumption}{Assumption}
\newtheorem{theorem}{Theorem}
\newtheorem{lemma}{Lemma}

\newtheorem{remark}{Remark}

\DeclareMathOperator{\diag}{diag}
\DeclareMathOperator{\SO}{SO}

\usepackage{color}

\AtBeginDocument{\setlength{\parskip}{0.5ex}} 
\AtBeginDocument{%
    \setlength\abovedisplayskip{6pt plus 2pt minus 2pt}%
    \setlength\belowdisplayskip{6pt plus 2pt minus 2pt}%
    \setlength\abovedisplayshortskip{0pt plus 2pt}%
    \setlength\belowdisplayshortskip{3pt plus 1pt minus 1pt}%
}

\begin{document}

\begin{frontmatter}

\title{Laplacian Spectral Shaping for Non-Uniform Scaling Formation Control of Open Multi-Agent Systems\thanksref{footnoteinfo}} 

\thanks[footnoteinfo]{A preliminary version of this work \cite{He2026} has been accepted for presentation at the IFAC World Congress 2026, where only agent joining as a follower was addressed. In contrast, this work contributes a unified framework covering all four elementary topology changes—agent joining, edge addition, agent leaving, and edge removal—and supports leader reassignment. Corresponding author Gangshan Jing.}

\author[Chongqing]{Tao He}\ead{20231301010@stu.cqu.edu.cn},    
\author[Chongqing]{Gangshan Jing}\ead{jinggangshan@cqu.edu.cn}              

\address[Chongqing]{School of Automation, Chongqing University, 
Chongqing, 400044, PRC}


\begin{abstract}                          
Non-uniform scaling control enables a multi-agent formation to adjust its shape by compressing or stretching independently along different coordinate axes through inter-agent interactions, offering high flexibility in complex environments. The fundamental idea is encoding the desired formation shape as the kernel of a matrix-valued Laplacian. In open multi-agent systems, however, changes in number of agents, number of edges, and leader selection dynamically alter this Laplacian, destroying the required spectral properties—positive semidefiniteness, correct kernel, and positive definiteness of the follower block (we summarize these properties as the formation spectrum). In this paper, we develop distributed protocols to strategically adjust partial weights of the Laplacian matrix for formation control in arbitrary dimensional space. By implementing the protocols, the desired formation spectrum can be preserved under dynamic topology changes including agent joining, edge addition, agent leaving, and edge removal, while any pair of agents can serve as leaders. Unlike existing Laplacian design methods for affine formation control under topology changes, the proposed approach requires a sparser sensing graph, avoids a predefined parent-child hierarchical structure, and supports leader reassignment. The effectiveness of the proposed protocols is validated through both theoretical analysis and numerical simulations.
\end{abstract}

\end{frontmatter}

\section{Introduction}
Distributed formation maneuver control enables collective shape transformations using only inter-agent local sensing, and has underpinned a wide range of applications, including high-speed cooperative payload transport \cite{Sun2025}, navigation through densely cluttered spaces \cite{Zhou2022}, and self-organizing heterogeneous robot teams for complex missions \cite{Zhu2024}. Critically, the realization of such distributed cooperative control inherently depends on the underlying sensing graph.

In realistic deployments, however, both edges and nodes of the sensing graph may change dynamically \cite{Pelin2025}. On the edge side, obstacle avoidance continuously perturbs inter-agent geometry, and sensing links may be intermittently blocked by occlusions or limited range, forcing edge removal or reconfiguration. Meanwhile, strategically adding new edges can accelerate synchronization convergence \cite{Cao2025} while simultaneously providing redundant information pathways to enhance robustness \cite{Trinh2020}. On the node side, agents may join to augment mission capability or leave due to failure or energy depletion, resulting in a fully dynamic vertex set \cite{Franceschelli2021,Jia2026}.These scenarios give rise to an \emph{open multi-agent system}, where the graph is neither static nor confined to predefined switching patterns \cite{Bo2008,OH2015}, but may undergo arbitrary edge addition/removal and agent joining/leaving events (we refer to this as fully dynamic topology changes).

Existing formation maneuver control methods can be divided into two paradigms, as systematically compared in Table~\ref{tab:formation-comparison}. The first uses nonlinear constraints such as distances \cite{Yu2006,Trinh2024}, bearings \cite{Zhao2016, Karimian2017,Trinh2020,Zhang2024}, or distance ratios/angles \cite{Jing2019,Chen2021,Buckley2021,Cao2020}. Built on minimally rigid graphs via Henneberg constructions, these methods naturally support agent joining and leaving, but they do not handle edge addition or removal, and their nonlinear controllers complicate stability analysis.

\begin{table*}[!thbp]
\scriptsize
\centering
\setlength{\tabcolsep}{4pt}
\begin{threeparttable}
\caption{Comparison of Formation Maneuver Control Methods in $\mathbb{R}^d$}
\label{tab:formation-comparison}
\begin{tabular}{l|c|c|c}
\hline
\textbf{Method} & \textbf{Maneuverability} (DoF) & \multicolumn{2}{c}{\textbf{Graph}\tnote{a}} \\
\hline
Position-based \cite{OH2015} & None ($0$) & Star & Static \\
Displacement-based \cite{Zhu2019,Cao2025} & Translation (d) & Connected & Fully dynamic \\
Distance-based \cite{Yu2006,Trinh2024} & Rotation-Translation ($\binom{d}{2}+d$) & Laman & Partially dynamic \\
Bearing-based \cite{Zhao2016, Karimian2017,Trinh2020,Zhang2024} & Uniform Scaling-Translation (d+1)  & Laman & Partially dynamic \\
Angle/Distance-ratio-based \cite{Jing2019,Chen2021,Buckley2021,Cao2020} & Uniform Scaling-Rotation-Translation ($\binom{d}{2}+d+1$) & Laman & Partially dynamic \\
Complex-Laplacian-based \cite{Lin2014,Hector2021,Fang2024} & Uniform Scaling-Rotation-Translation (4 in $\mathbb{R}^2$) & 2-rooted & Static \\
Affine Formation \cite{Lin2016,Zhao2018,Li2025} & Affine $(d^2+d)$ & (d+1)-lateration & Partially dynamic \\
Our Work & Non-uniform Scaling-Translation (2d)  & 2-vertex-connected & Fully dynamic \\
\hline
\end{tabular}

\begin{tablenotes}[flushleft]
\scriptsize
\item[a] Fully dynamic: supports all four primitive topology changes—agent joining, agent leaving, edge addition, and edge removal. Partially dynamic: supports only a subset of the primitive topology changes. Static: no changes are supported.
\end{tablenotes}
\end{threeparttable}
\end{table*}

The linear Laplacian-based paradigm encodes the target formation in the Laplacian null space, yielding linear controllers and facilitating global convergence analysis. Within this paradigm, displacement-based methods \cite{Zhu2019,Cao2025} require only a connected graph and allow independent edge-weight adjustments; consequently, they support fully dynamic topology changes, but achieve only translation ($d$ degrees of freedom, DoF). Complex-Laplacian-based methods \cite{Lin2014,Hector2021,Fang2024} provide the additional ability to uniformly scale the formation, thereby enhancing maneuverability. However, this isotropic scaling mechanism restricts their adaptability to anisotropic environments. Consequently, they may compromise safety or efficiency by reducing scale in directions where it is unnecessary. This motivates the need for non-uniform scaling (independent scaling along each axis), which is supported by both affine formation methods \cite{Lin2016,Zhao2018,Li2025} and the matrix-valued constraint framework \cite{he2025}. However, the former, which handles topology changes \cite{Li2025}, relies on a predefined hierarchical $(d+1)$-rooted graph structure, lacks mechanisms for edge addition, and offers no structural flexibility to handle leader faults. The latter achieves non-uniform scaling under merely $2$-rooted graphs with two leaders, but lacks mechanisms for dynamic topology changes, requires solving a centralized high-dimensional inverse eigenvalue problem for global convergence, and fixes the leader set a priori. 

This paper addresses these limitations by developing a fully distributed framework for non-uniform scaling formation maneuver control that operates under arbitrary topology changes. The main contributions are summarized as follows. (i) We introduce a unified algebraic representation for the Laplacian dynamics of open multi-agent systems. Based on this formulation, we develop distributed update laws that ensure maintainance of the required spectral property of the Laplacian during topology changes, including agent joining, agent leaving, edge addition, and edge removal. (ii) We adopt a 2-vertex-connected graph structure to enable non-uniform scaling formation maneuver control in arbitrary dimensional space. In contrast to existing formation control methods \cite{Lin2014,Zhao2018,Hector2021,Fang2024,Li2025} that require a predetermined and immutable leader set, the proposed framework supports dynamic leader reassignment. Moreover, compared with existing affine formation frameworks under topology changes (e.g., \cite{Li2025}), the proposed approach requires a sparser sensing graph, eliminates the need for predefined parent–child hierarchical structure, and provides a characterization of minimal compensatory edge sets for restoring $2$-vertex-connectivity.


The remainder of this paper is organized as follows. Section~\ref{sec:preliminaries} introduces mathematical background and fundamentals of formation control. Section~\ref{sec:problem} states the problem. Section~\ref{sec:topology} details the spectral property maintenance strategy under topology changes. Section~\ref{sec:simulations} presents numerical results and Section~\ref{sec:conclusion}  concludes the paper and discusses future work.

\section{Preliminaries}\label{sec:preliminaries}
\subsection{Notation and Graph Theory}
Throughout this paper, $\mathbb{R}$ denotes the field of real numbers. $\mathbb{R}^d$ signifies the $d$-dimensional Euclidean space. The identity matrix is $ I_n \in \mathbb{R}^{n \times n} $, the all-ones vector is $ 1_n \in \mathbb{R}^n $, the zero tensor (scalar/vector/matrix) with context-appropriate dimensions is $ 0 $, the Kronecker product is $ \otimes $, and the symbols $ \succeq $ and $ \succ $ denote positive semi-definiteness and positive definiteness for matrices, respectively. The Special Orthogonal group is $\SO(d) = \{ R \in \mathbb{R}^{d \times d} \mid R^\top R = I_d, \det(R) = 1 \}$.

We model the interaction network of a multi-agent system in $\mathbb{R}^d$ as an undirected graph $G= (V, E)$. The vertex set $V=\{1,\dots,n\}$ represents agents. The edge set $E\subseteq \{ \{u, v\} \mid u, v \in V, u \neq v \}$ captures both sensing and communication links: an edge \( \{u, v\} \in E\) indicates that agents \(u\) and \(v\) can sense or exchange information with each other. The neighborhood of agent $i$ is defined as $N_i=\{\, j \in V : \{i, j\} \in E \,\}$. 

A matrix $L = [L_{ij}] \in \mathbb{R}^{dn \times dn}$ is called a \emph{matrix-valued Laplacian} associated with an undirected graph $G$ with $n$ agents in $\mathbb{R}^d$ if it satisfies $\sum_{j=1}^n L_{ij} = 0$ for all $i = 1, \dots, n$, where for $i \neq j$, each block $L_{ij} \in \mathbb{R}^{d \times d}$ (with entries possibly positive, negative, or zero) represents the matrix weight of pair $\{i,j\}$ (nonzero if $\{i,j\} \in E$, and zero otherwise), and satisfies $L_{ij} = L_{ji}^\top$.

For a graph $G$, let \(G - v\) (resp. \(G - e\))  be the subgraph of \(G\) obtained by deleting a vertex \(v\) and all its incident edges (resp. deleting an edge \(e\)).

\begin{definition}[2-vertex-connected \cite{diestel2017}]
\label{def:2vc}
A graph $G = (V,E)$ is 2-vertex-connected if $G - x$ is connected for all $x \in V$.
\end{definition}

\begin{definition}[2-rooted \cite{Lin2014}]
\label{def:2rooted}
A graph $G = (V,E)$ is called 2-rooted if there exist distinct vertices $r_1,r_2\in V$, called the roots, such that for every $w \in V \setminus \{r_1,r_2\}$ and every $x \in V \setminus \{w\}$, there exists a path from $\{r_1,r_2\}$ to $w$ in $G - x$.
\end{definition}

\begin{lemma}\label{lem:2vc}
If $G = (V,E)$ is 2-rooted with roots $\{r_1,r_2\}$, then $G' = (V, E')$ with $E' = E\cup\{\{r_1,r_2\}\}$ is 2-vertex-connected.
\end{lemma}

\begin{pf}
We prove this by two cases. \emph{Case 1: $x \notin \{r_1,r_2\}$.} 
Since edge $\{r_1,r_2\} \in E'$, $r_1$ and $r_2$ remain connected in $G'-x$. For any $w \in V \setminus \{r_1,r_2,x\}$, the 2-rooted property of $G$ guarantees a path from $\{r_1, r_2\}$ to $w$ in $G'-x$. Thus $G'-x$ is connected. \emph{Case 2: $x = r_1$ (symmetric for $x = r_2$).} For any $w \in V \setminus \{r_1, r_2\}$, the 2-rooted property of $G$ gives a path from $r_2$ to $w$ in $G-r_1$. Hence $G'-x$ is connected.
\end{pf}

\begin{lemma}\label{lem:2vc_2root}
If $G = (V,E)$ is 2-vertex-connected, then $G$ is 2-rooted, and any pair of distinct vertices \( \{u, v\} \subseteq V \) can serve as roots.
\end{lemma}
\begin{pf}
Fix any pair \(\{u,v\}\subseteq V\). For any \(w\in V\setminus\{u,v\}\) and any \(x\in V\setminus\{w\}\), at least one of \(u,v\) lies in \(G-x\); the connectivity of \(G-x\) then yields a path from that root to \(w\) in \(G-x\). Thus $G$ is 2-rooted with roots \(\{u,v\}\).
\end{pf}

\subsection{Formation and Desired Shape}
A configuration of the agents in $V$ is denoted as $(p, R)$, where $p = [p_1^\top, \dots, p_n^\top]^\top \in \mathbb{R}^{dn}$ is the stacked vector of agent positions, and $R \in \SO(d)$ defines the axes along which non-uniform scaling transformation is applied. Consequently, a formation in $\mathbb{R}^d$ is fully characterized by the tuple $(G, p, R)$. Let $(\tilde{p}, R)$ be a nominal configuration, where $\tilde{p}= [\tilde{p}_1^\top, \dots, \tilde{p}_n^\top]^\top \in \mathbb{R}^{dn}$ is the nominal position vector that defines the reference shape of the formation. The desired shape manifold under non-uniform scaling is
\begin{equation} \label{eq_dss}
    \varPi(\tilde p,R) \coloneqq \left\{ p = \left(I_n \otimes S(R)  \right) \tilde{p} + 1_n \otimes \tau \mid s, \tau \in \mathbb{R}^d \right\},
\end{equation}
where $S(R)=R \diag(s) R^{\top}$, $s \in \mathbb{R}^d$ and $\tau \in \mathbb{R}^d$ denote the scaling and translation maneuver parameters, respectively.

This formulation generalizes uniform scaling by allowing independent compression or stretching along orthogonal directions defined by $R \in \SO(d)$. As illustrated in Fig.~\ref{fig:scaling_motivation}, in narrow passages, the formation can compress along the restricted direction while maintaining adequate spacing in others, avoiding the excessive proximity and inefficient maneuvers associated with uniform scaling approaches \cite{Lin2014, Jing2019, Trinh2020}.

The following lemma characterizes the dimension of the desired shape manifold.
\begin{lemma}[\cite{he2025}]
    \label{lem:dimVARPi}
    For a nominal formation $(G, \tilde{p}, R)$, $\dim(\varPi(\tilde{p},R)) = 2d $ if and only if for each $l = 1,\dots,d$, the set $\{\tilde{p}^l_{i,R}\}_{i \in V}$ is not a singleton, where $\tilde{p}^l_{i,R}$ denotes the $l$-th coordinate of $R^{\top}\tilde{p}_i$.
\end{lemma}

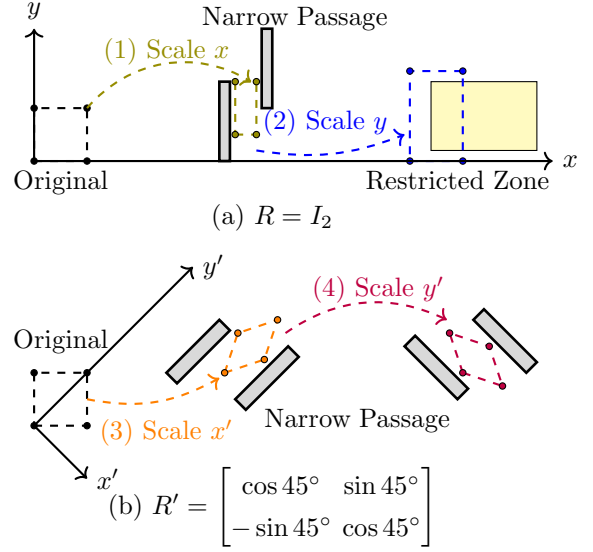
\begin{figure}[t]
	\centering
	\begin{tikzpicture}[scale=0.7]
	
	\begin{scope}

	\coordinate (A1) at (0,0);
	\coordinate (A2) at (1,0);
	\coordinate (A3) at (1,1);
	\coordinate (A4) at (0,1);
	
	\draw[dashed, thick] (A1) -- (A2) -- (A3) -- (A4) -- cycle;
	
	\foreach \point in {A1,A2,A3,A4} {
		\draw[fill=black] (\point) circle (0.06);
	}
	
	\draw[->, thick, olive, dashed, bend left=40] (1,1) to (4.1,1.5);
	
	\draw[fill=gray!30, line width=1pt] (3.5,0) -- (3.5,1.5) -- (3.7,1.5) -- (3.7,0) -- cycle;
	\draw[fill=gray!30, line width=1pt] (4.3,1) -- (4.3,2.5) -- (4.5,2.5) -- (4.5,1) -- cycle;
	\node[right] at (3,2.7) {Narrow Passage};

	\coordinate (C1) at (3.8,0.5);
	\coordinate (C2) at (4.2,0.5);
	\coordinate (C3) at (4.2,1.5);
	\coordinate (C4) at (3.8,1.5);
	
	\draw[dashed, thick, olive] (C1) -- (C2) -- (C3) -- (C4) -- cycle;
	
	\foreach \point in {C1,C2,C3,C4} {
		\draw[fill=olive] (\point) circle (0.06);
	}
	
	\draw[fill=yellow!30] (6.5+1,0.2) -- (8.5+1,0.2) -- (8.5+1,1.5) -- (6.5+1,1.5) -- cycle;
	\node[below] at (8,0) {Restricted Zone};
	
	\draw[->, thick, blue, dashed, bend right=15] (4.2,0.2) to (6+1,0.5);
	
	\coordinate (B1) at (6.1+1,0);
	\coordinate (B2) at (7.1+1,0);
	\coordinate (B3) at (7.1+1,1.7);
	\coordinate (B4) at (6.1+1,1.7);
	
	\draw[dashed, thick, blue] (B1) -- (B2) -- (B3) -- (B4) -- cycle;
	
	\foreach \point in {B1,B2,B3,B4} {
		\draw[fill=blue] (\point) circle (0.06);
	}
	
	\draw[->, thick] (0,0) -- (9.8,0) node[right] {$x$};
	\draw[->, thick] (0,0) -- (0,2.5) node[above] {$y$};
	
	\node[below] at (0.5,0) {Original};
	\node[below, olive] at (2.5,2.5) {(1) Scale $x$};
	\node[above, blue] at (5.5,0.3) {(2) Scale $y$};
	\end{scope}
	
	\node[above] at (4.5,-1.5) {(a) $R = I_2$};
		
	\begin{scope}[shift={(0,-5)}]
	
	\draw[->, thick] (0,0) -- (3,3) node[right] {$y'$};
	\draw[->, thick] (0,0) -- (1,-1) node[right] {$x'$};
	
	\coordinate (A1) at (0,0);
	\coordinate (A2) at (1,0);
	\coordinate (A3) at (1,1);
	\coordinate (A4) at (0,1);
	
	\draw[dashed, thick] (A1) -- (A2) -- (A3) -- (A4) -- cycle;
	
	\foreach \point in {A1,A2,A3,A4} {
		\draw[fill=black] (\point) circle (0.06);
	}
	
	\draw[fill=gray!30, line width=1pt] (2.5,1.0) -- (3.5,2.0) -- (3.7,1.8) -- (2.7,0.8) -- cycle;
	\draw[fill=gray!30, line width=1pt] (3.8,0.5) -- (4.8,1.5) -- (5.0,1.3) -- (4.0,0.3) -- cycle;

	\coordinate (D1) at (0.6+3.0,1.0); 
	\coordinate (D2) at (0.6+3.75,1.25);
	\coordinate (D3) at (0.6+4.0,2.0);
	\coordinate (D4) at (0.6+3.25,1.75);
	
	\draw[dashed, thick, orange] (D1) -- (D2) -- (D3) -- (D4) -- cycle;
	
	\foreach \point in {D1,D2,D3,D4} {
		\draw[fill=orange] (\point) circle (0.06);
	}
	
	\draw[->, thick, orange, dashed, bend right=20] (1,0.5) to (3.5,0.9);
	
	\draw[fill=gray!30, line width=1pt] (5.0+2,1.5) -- (6.0+2,0.5) -- (6.2+2,0.7) -- (5.2+2,1.7) -- cycle;
	\draw[fill=gray!30, line width=1pt] (5.3+3,2) -- (6.3+3,1) -- (6.5+3,1.2) -- (5.5+3,2.2) -- cycle;
	\node[below] at (6.1,0.5) {Narrow Passage};
	
	\coordinate (E1) at (2.6+5.5,1.5-0.5);   
	\coordinate (E2) at (2.6+6.25,1.25-0.5); 
	\coordinate (E3) at (2.6+6.0,2.0-0.5);   
	\coordinate (E4) at (2.6+5.25,2.25-0.5); 
	
	\draw[dashed, thick, purple] (E1) -- (E2) -- (E3) -- (E4) -- cycle;
	
	\foreach \point in {E1,E2,E3,E4} {
		\draw[fill=purple] (\point) circle (0.06);
	}
	
	\draw[->, thick, purple, dashed, bend left=30] (4.75,1.75) to (7.8,2);
	
	\node[above] at (0.5,1.2) {Original};
	\node[below, orange] at (2.5,0.3) {(3) Scale $x'$};
	\node[below, purple] at (6.5,3) {(4) Scale $y'$};
	\node[above] at (4.5,-2.5) {(b) $R' = \begin{bmatrix} \cos 45^\circ & \sin 45^\circ \\ -\sin 45^\circ & \cos 45^\circ \end{bmatrix}$};
	\end{scope}
	
	\end{tikzpicture}
	\caption{Non-uniform scaling transformations of the formation for obstacle avoidance along the axes defined by $R$ and $R'$, respectively.}
	\label{fig:scaling_motivation}
\end{figure}

\subsection{Agent Dynamics and Shape Stabilization}
Each agent \(i \in V\) is governed by single-integrator dynamics
\begin{equation}\label{eq_dy}
    \dot{p}_i = u_i,
\end{equation}
where \(p_i \in \mathbb{R}^d\) is its position and \(u_i \in \mathbb{R}^d\) is the control input. Since only distributed control is considered, agent \(i\) has access to solely relative information with respect to its neighbors \(N_i\): $\{p_j-p_i,j\in N_i\}$.

To drive the configuration \(p\) toward the desired manifold \(\varPi(\tilde{p},R)\), a commonly used distributed control law takes the form
\begin{equation}
    u_i = -\sum_{j \in N_i} L_{ij} (p_j - p_i),
\end{equation}
where the matrix weights \(L_{ij} \in \mathbb{R}^{d \times d}\) (to be designed) act on the relative positions. In stacked form, the closed-loop system becomes
\begin{equation}
    \dot{p} = -L p,
\end{equation}
where \(L\) is a matrix-valued Laplacian. If, in addition, $\ker(L) = \varPi(\tilde{p},R)$ and $L \succeq 0$, then using the Lyapunov function $V=\frac12\|p\|^2$, it can be proved that $p$ converges to the desired manifold $\varPi(\tilde{p},R)$ asymptotically.

\subsection{Formation Maneuvering and Leader Reconfiguration}
To enable the formation to actively execute translations and non-uniform scalings of the nominal shape, we adopt a leader-follower strategy. Specifically, a small set of agents, designated as leaders, are directly controlled to realize the desired maneuver, while the remaining followers, through the distributed control law, automatically adjust their positions to maintain the formation.

Unlike the shape stabilization problem discussed in the previous subsection, formation maneuvering imposes additional requirements on the matrix-valued Laplacian. Given a leader set $V_l \subset V$ with follower set $V_f = V \setminus V_l$, we have $p = [p_l^\top, p_f^\top]^\top$ and $
L = [\begin{smallmatrix} L_{ll} & L_{lf} \\ L_{fl} & L_{ff} \end{smallmatrix}]
$, accordingly. If $\ker(L) = \varPi(\tilde{p},R)$, then $L p = 0$ for any $p \in \varPi(\tilde{p},R)$. If the block $L_{ff}$ is nonsingular, followers are uniquely determined by leaders:
\begin{equation} \label{eq_flr}
p_f = - (L_{ff})^{-1} L_{fl} p_l.
\end{equation}

In contrast to \cite{he2025} where the leader set is fixed a priori, we further require that the leaders can be dynamically selected from the agents to accommodate leader failures or task-driven reconfigurations. Therefore, we introduce the following notion.

\begin{definition}[Formation Spectrum]
\label{def:formation_spectrum}
For a nominal formation $(G,\tilde{p},R)$ with $n$ agents in $\mathbb{R}^d$, the matrix-valued Laplacian $L \in \mathbb{R}^{dn \times dn}$ is said to have a formation spectrum if it satisfies the following three conditions:
\begin{enumerate}
    \item[(i)] $L \succeq 0$;
    \item[(ii)] $\ker(L) = \varPi(\tilde{p},R)$;
    \item[(iii)] $Q^\top L Q \succ 0,\ \ \forall Q \in \mathcal{Q}$
\end{enumerate}
where $\mathcal{Q} = \{ Q \otimes I_d \mid Q \in \{0,1\}^{n \times (n-2)},\ Q^\top Q = I_{n-2} \}$ selects any two agents as leaders, and $Q^\top L Q$ is the corresponding follower block.
\end{definition}

The constraint $Q^\top LQ\succ0$ for $\forall Q\in\mathcal{Q}$ ensures that $L_{ff}\succ0$ always holds no matter which two agents are selected as leaders. Such setting benefits for enhancing the adaptability of the formation to complex environments where leaders may need to be rechosen \cite{Li2024}.

By Lemma~\ref{lem:dimVARPi}, the following assumption ensures that $\dim(\varPi(\tilde{p}_l,R)) = 2d$ for any two leaders, thereby establishing a one-to-one correspondence between the leader state and the maneuver parameters $s$ and $\tau$.

\begin{assumption}\label{ass:graph}
For the nominal formation $(G,\tilde{p},R)$, the configuration $(\tilde{p}, R)$ satisfies $\prod_{l=1}^d \tilde{p}^l_{ji,R} \neq 0$ for every pair of distinct agents $j,i$, where $\tilde{p}^l_{ji,R}$ denotes the $l$-th coordinate of $R^{\top}(\tilde{p}_j - \tilde{p}_i)$.
\end{assumption}

Assumption~\ref{ass:graph} is readily satisfied. Given any set of distinct nominal positions $\{\tilde{p}_i\}$, there always exists a rotation matrix $R \in \SO(d)$ such that Assumption~\ref{ass:graph} holds. To see this, observe that for each pair of distinct indices $\{j,i\}$ and each coordinate $l \in \{1,\dots,d\}$, the equation $\tilde{p}^l_{ji,R} = 0$ defines a proper algebraic subset of $\SO(d)$ of codimension at least one. Taking the union over the finitely many such pairs and coordinates yields a closed set $\mathcal{F} \subset \SO(d)$ which is therefore a finite union of lower-dimensional algebraic sets. Hence its complement $\SO(d) \setminus \mathcal{F}$ is nonempty. Choosing any $R$ from this complement guarantees that Assumption~\ref{ass:graph} holds.

\begin{remark}\label{rem:design_feasibility1}
Compared to classical affine formation control~\cite{Zhao2018}, which requires at least $d+1$ leaders in $\mathbb{R}^d$, the proposed approach needs only two leaders regardless of the dimension $d$. Moreover, our approach requires only that the nominal positions $\{\tilde{p}_i\}$ be pairwise distinct, avoiding the generic position assumptions (e.g., no three collinear in 2-D, no four coplanar in 3-D) commonly assumed in  distance-ratio-based \cite{Cao2020}, angle-based \cite{Buckley2021}, affine \cite{Zhao2018}, clique-based \cite{HeGen2025}, and bearing-based \cite{Erskine2024} methods.
\end{remark}

The following lemma provides a necessary graph-theoretic condition for the formation spectrum.

\begin{lemma}\label{lem:2vc_sp}
If the nominal formation \((G,\tilde p,R)\) satisfies Assumption~\ref{ass:graph} and has a Laplacian \(L\) with a formation spectrum, then \(G\) is \(2\)-vertex-connected.
\end{lemma}

\begin{pf}
Suppose \(G\) is not \(2\)-vertex-connected. Then there exists a cut vertex \(x\) such that \(G-x\) has at least two connected components. Denote the vertex sets of these components by \(V_1,\dots,V_k\) with \(k\ge2\). We consider two cases.

\textbf{Case 1:} There exists a component \(V_i\) with \(|V_i| \ge 2\). Choose two distinct agents \(r_1, r_2 \in V_i\) as leaders. Let \(U = \bigcup_{j \neq i} V_j\) and \(\bar{U} = V \setminus (U \cup \{x\})\). Reorder the vertices so that the Laplacian \(L\) has the block form
\begin{equation}
L = \begin{bmatrix}
L_{\bar{U}\bar{U}} & L_{\bar{U}x} & 0      \\
L_{x\bar{U}}          & L_{xx}          & L_{xU} \\
0                         & L_{Ux}          & L_{UU}
\end{bmatrix}.
\end{equation}
Set \(p_\beta = (p_x^\top, p_U^\top)^\top\). By Lemma \ref{lem:dimVARPi} and Assumption~\ref{ass:graph}, we have \(\dim(\Pi(\tilde p_\beta,R)) = 2d\). For any \(p_\beta \in \Pi(\tilde p_\beta,R)\), the equation \(Lp=0\) restricted to \(U\) gives
\begin{equation}
[L_{Ux} , L_{UU}]p_{\beta} = 0.
\end{equation}
Hence the rows of \([L_{Ux}\; L_{UU}]\) are linearly dependent. These rows are part of the follower block \(L_{ff}\), so \(L_{ff}\) is singular—a contradiction.

\textbf{Case 2}: All components are singletons. Choose leaders as two leaves. For any other leaf \(y\), \(L_{yx}p_x+L_{yy}p_y=0\) for all \(p\in\Pi\). With \(p=\mathbf{1}_n\otimes\tau\) we get \(L_{yx}+L_{yy}=0\), so \(L_{yx}(p_x-p_y)=0\). By Assumption~\ref{ass:graph}, \(p_x-p_y\) spans \(\mathbb{R}^d\). Hence \(L_{yx}=0\) and \(L_{yy}=0\). Thus rows of \(L_{ff}\) indexed by leaves are zero, making \(L_{ff}\) singular—a contradiction.

Therefore \(G\) is \(2\)-vertex-connected.
\end{pf}

\begin{remark}
Lemma~\ref{lem:2vc_sp} shows that the formation spectrum necessarily requires the underlying graph to be $2$-vertex-connected. This condition is stronger than the $2$-rooted graphs used in prior works \cite{Lin2014,he2025}, but is nonetheless practical: any $2$-rooted graph becomes $2$-vertex-connected by adding an edge between its two roots (Lemma~\ref{lem:2vc}); in practice, such an edge naturally exists because the two leaders typically need to communicate. Moreover, $2$-vertex-connected graphs allow any pair of agents to serve as leaders (Lemma~\ref{lem:2vc_2root}), enabling dynamic leader reconfiguration—capabilities essential for open multi-agent systems.
\end{remark}


\section{Problem Statement} \label{sec:problem}
We consider an open multi-agent system in which agents and interaction links may dynamically appear or disappear. The interaction topology evolves through one of four elementary operations at each step: agent joining, agent leaving, edge addition, or edge removal. We index the system after each such elementary operation by $t\in\mathbb{N}$. At stage $t$, the system is characterized by a nominal formation $F^t = (G^t, \tilde{p}^t, R)$, where $G^t= (V^t, E^t)$.

Our goal is to construct, for each $F^t$, a matrix-valued Laplacian $L^t\in\mathbb{R}^{dn^t\times dn^t}$ that has a formation spectrum. To avoid centralized recomputation after each topological change, we seek a recursive update mechanism: when an elementary operation transforms $F^t$ into $F^{t+1} = (G^{t+1}, \tilde{p}^{t+1}, R)$ with $G^{t+1} = (V^{t+1}, E^{t+1})$, the Laplacian should be updated to $L^{t+1}$ in a local manner with minimum structural perturbation. Any such update can be expressed in the general form
\begin{equation}\label{eq:L_update}
L^{t+1} = \mathcal{E}^t \bigl(L^t + \Delta^t\bigr) + \bar{\Delta}^t,
\end{equation}
where $\mathcal{E}^t: \mathbb{R}^{dn^t \times dn^t} \to \mathbb{R}^{dn^{t+1} \times dn^{t+1}}$ is the dimension-adjustment operator for agent joining/leaving, while $\Delta^t=[\Delta^t_{ij}]$ and $\bar{\Delta}^t = [\bar{\Delta}_{ij}^t]$ encode local edge weight modifications applied before and after $\mathcal{E}^t$, respectively, where $\Delta^t_{ij} \in \mathbb{R}^{d\times d}$ denotes the $(i,j)$-th block of $\Delta^t$ (and similarly for $\bar{\Delta}^t$). Associated with each update, we define
\begin{equation}
    E_{\text{add}} = \{ \{i,j\} \notin E^t \,|\, \Delta^t_{ij} \neq 0 \text{ or } \bar{\Delta}^t_{ij} \neq 0 \}
\end{equation}
as the set of edges added to the graph, and
\begin{equation}
    E_{\text{mod}} = \{ \{i,j\} \in E^t \,|\, \Delta^t_{ij} \neq 0 \text{ or } \bar{\Delta}^t_{ij} \neq 0 \}
\end{equation}
as the set of existing edges whose weights are modified. 

The challenge is to design these local updates so that $L^{t+1}$ inherits the desired spectral properties from $L^t$, yet any naive local modification to the graph may disrupt these global spectral properties of the Laplacian matrix, as illustrated in Fig.~\ref{fig:idea}. This leads to the following problem.

\begin{figure}[tbp]
	\centering
        \includegraphics[width=\linewidth]{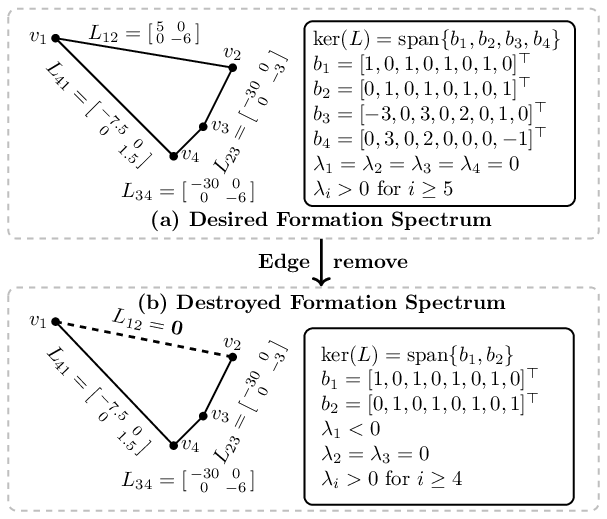}
        \caption{Effect of a simple edge removal on the formation spectrum of the matrix-valued Laplacian $L$. (a) $\ker(L)$ has four basis vectors $b_i$, all nonzero eigenvalues $\lambda_i$ are positive. (b) After deleting edge $\{v_1,v_2\}$ (its weight set to $0$), $\ker(L)$ reduces to two basis vectors only, a negative eigenvalue appears. This raises the question: how to preserve the formation spectrum under dynamic topology changes?}
	\label{fig:idea}
\end{figure}

\begin{problem}
\label{pbm:topology_update}
Given a nominal formation $F^t$ under Assumption~\ref{ass:graph} and its Laplacian $L^t$ possessing a formation spectrum. For each elementary operation (i)-(iv), design $\Delta^t$ and $\bar{\Delta}^t$ in \eqref{eq:L_update} such that the induced $F^{t+1}$ has a Laplacian $L^{t+1}$ that also possesses a formation spectrum, and among all such designs, $|E_{\text{mod}}|$ is minimized subject to the constraint that $|E_{\text{add}}|$ is minimized.
\begin{enumerate}
\item \textbf{Agent joining:} Add a new agent $v\notin V^t$ with prescribed nominal position $\tilde p_v\in\mathbb R^d$, yielding 
$\tilde{p}^{t+1}=[(\tilde{p}^t)^{\top},\tilde{p}_v^{\top}]^{\top}$, 
$V^{t+1}=V^t\cup \{v\}$, and $\mathcal{E}^t(\mathcal{M}) = [\begin{smallmatrix} \mathcal{M} & 0 \\ 0 & 0 \end{smallmatrix}]$.

\item \textbf{Edge addition:} Add a new edge $\{j,k\} \notin E^t$, yielding 
$\tilde{p}^{t+1}=\tilde{p}^t$, 
$V^{t+1}=V^t$,
 $\mathcal{E}^t(\mathcal{M}) = \mathcal{M}$, and $\{j,k\} \in E^{t+1}$.

\item \textbf{Agent leaving:} Remove an agent $u\in V^t$ and all its incident edges, yielding 
$\tilde{p}^{t+1}=\tilde{p}^t|_{V^{t+1}}$ (subvector indexed by $V^{t+1}=V^t\setminus\{u\}$), 
$\mathcal{E}^t(\mathcal{M}) = \mathcal{M}_{V^{t+1},V^{t+1}}$ (the principal submatrix indexed by $V^{t+1}$), and $\forall i\in V^t,\; \{u,i\}\notin E^{t+1}$.

\item \textbf{Edge removal:} Remove an existing edge $\{j,k\} \in E^t$, yielding 
$\tilde{p}^{t+1}=\tilde{p}^t$,
$V^{t+1}=V^t$,
$\mathcal{E}^t(\mathcal{M}) = \mathcal{M}$, and $\{j,k\} \notin E^{t+1}$.
\end{enumerate}
\end{problem}
We prioritize minimizing $|E_{\mathrm{add}}|$ since edge additions correspond to structural expansions of the interaction network, potentially requiring new sensing or communication links and thus incurring higher implementation cost. In contrast, edge modifications operate on existing interactions by adjusting edge weights without changing the underlying topology.

\section{Spectral Property Maintenance under Topology Changes}\label{sec:topology}
To maintain the spectral properties under topology changes in a distributed manner, we first introduce a matrix-weighted linear constraint from \cite{he2025}. For any agent $k$ and any two of its neighbors $i,j\in N_k$, the constraint is given by
\begin{equation}
W_{jk} p_{ik}+W_{ki} p_{jk}=0,
\label{equ_vvc}
\end{equation}
where $p_{ik}=p_i-p_k$, $p_{jk}=p_j-p_k$, $W_{jk} = \diag\left(\tilde{p}_{jk,R}\right) R^{\top}$, $W_{ki} = \diag\left(\tilde{p}_{ki,R}\right) R^{\top}$, $\tilde{p}_{ki,R}=R^{\top} \tilde{p}_{ki}$, $\tilde{p}_{ki} = \tilde{p}_k - \tilde{p}_i$, $\tilde{p}_{jk,R}=R^{\top} \tilde{p}_{jk}$, $\tilde{p}_{jk} = \tilde{p}_j - \tilde{p}_k$. 

\begin{lemma}[\cite{he2025}] \label{lem:2}
The constraint \eqref{equ_vvc} is invariant under the non-uniform scaling and translation transformations defined in \eqref{eq_dss}. That is, if \((p_i, p_j, p_k)\) satisfies \eqref{equ_vvc}, then for any \(s, \tau \in \mathbb{R}^d\), the transformed triplet $(S(R)p_i+\tau,\; S(R)p_j+\tau,\; S(R)p_k+\tau)$ also satisfies \eqref{equ_vvc}, where \(S(R)=R\operatorname{diag}(s)R^\top\) as in \eqref{eq_dss}.
\end{lemma}

Lemma \ref{lem:2} implies that if we construct a Laplacian matrix $L$ based on \eqref{equ_vvc}, it always holds that $\Pi(\tilde{p},R)\subseteq\ker(L)$, which lays the foundation for establishing a formation spectrum. We now solve the four primitive topology operations in Problem~\ref{pbm:topology_update} under this matrix-weighted linear constraint. For notational simplicity, we use $F$, $L$, $F^+$, $L^+$, $\mathcal{E}$, $\Delta$, and $\bar{\Delta}$ to denote $F^t$, $L^t$, $F^{t+1}$, $L^{t+1}$, $\mathcal{E}^t$, $\Delta^t$, and $\bar{\Delta}^t$, respectively.

\subsection{Agent Joining} \label{subsec:agent_joining}
We now address agent joining in Problem~\ref{pbm:topology_update}. Let $v\notin V$ be the new agent with a pre-assigned nominal position $\tilde{p}_v$.  
To integrate $v$ into the formation while preserving the spectral properties, we connect it to two existing neighboring agents $i,j\in V$ such that $\{i,j\}\in E$, and set $E_{\text{add}} = \{\{i,v\},\{j,v\}\}$.

We next detail the construction of the Laplacian matrix for the updated formation. For the nominal formation $(G, \tilde{p}, R)$, we first reorder the agents in $V$ to place $i$ and $j$ at the end. The corresponding Laplacian matrix $L$ can be partitioned as
\begin{equation}
L = \begin{bmatrix}
L_{rr} & L_{rs} \\
L_{sr} & L_{ss}
\end{bmatrix},
\label{eq:L_block}
\end{equation}
where subscript $r$ corresponds to agents in $V \setminus \{i,j\}$ and $s$ corresponds to agents in $\{i,j\}$. For the triplet $(i,j,v)$, the constraint derived from \eqref{equ_vvc} takes the matrix form
\begin{equation}
M_{ijv} p_{ijv} = 0,
\label{eq:iju_matrix}
\end{equation}
where $M_{ijv} = \begin{bmatrix} W_{jv}, & W_{vi}, & W_{ij} \end{bmatrix}$, $W_{ij}=\diag\left(\tilde{p}_{ij,R}\right) R^{\top}$, $p_{ijv} = [p_i^\top, p_j^\top, p_v^\top]^\top$.

To embed this constraint distributively into the Laplacian framework, we define the positive semidefinite block
\begin{equation}\label{eq:ijv_block}
L^{ijv} = M_{ijv}^\top D M_{ijv} \in \mathbb{R}^{3d \times 3d},
\end{equation}
where $D\succ 0$ is a diagonal design matrix that provides a degree of freedom for weighting the constraint.

Following the update formula \eqref{eq:L_update}, we set $\Delta = 0$ and use $\mathcal{E}$ as defined in Problem~\ref{pbm:topology_update} for agent joining.  
The matrix $\bar{\Delta}$ is obtained by embedding $L^{ijv}$ into a $d|V^+| \times d|V^+|$ zero matrix according to the ordering $(r,i,j,v)$:
\begin{equation}\label{eq:Liju_pad}
\bar{\Delta} = L^{ijv}_{\text{pad}}(D) 
= \begin{bmatrix}
0 & 0 & 0 & 0 \\
0 & W_{jv}^\top D W_{jv} & W_{jv}^\top D W_{vi} & W_{jv}^\top D W_{ij} \\
0 & W_{vi}^\top D W_{jv} & W_{vi}^\top D W_{vi} & W_{vi}^\top D W_{ij} \\
0 & W_{ij}^\top D W_{jv} & W_{ij}^\top D W_{vi} & W_{ij}^\top D W_{ij}
\end{bmatrix}. 
\end{equation}

The Laplacian for the extended formation $(G^+,\tilde{p}^+,R)$ is
\begin{equation}
L^+ = \mathcal{E}(L) + L^{ijv}_{\text{pad}}(D).
\label{eq:LVA}
\end{equation}

\begin{theorem} \label{thm:agent_joining}
Let $F=(G,\tilde{p},R)$ satisfy Assumption~\ref{ass:graph} and $L$ be its Laplacian with a formation spectrum. For a new agent $v$ with nominal position $\tilde{p}_v$ such that the extended nominal configuration $(\tilde{p}^+ = [\tilde{p}^\top, \tilde{p}_v^\top]^\top,R)$ satisfies Assumption~\ref{ass:graph}, choose two distinct agents $i,j\in V$ with $\{i,j\}\in E$. Then there exists a diagonal matrix $D\succ0$ such that setting $\Delta=0$, $\bar{\Delta}$ given by \eqref{eq:Liju_pad} yields a formation $F^+ = (G^+,\tilde{p}^+,R)$ whose Laplacian $L^+$ given by \eqref{eq:LVA} has a formation spectrum. Moreover, $|E_{\text{add}}|=2$ and $|E_{\text{mod}}|=1$ are minimal.
\end{theorem}

\begin{pf}
    See Appendix \ref{sec:proofTheagent_joining}.
\end{pf}

Based on this construction, a fully distributed implementation is given in Algorithm~\ref{alg:agent_joining}. 


\begin{algorithm}[htbp]
\caption{Distributed Laplacian Update for Agent Joining}
\label{alg:agent_joining}
\begin{algorithmic}[1]
\REQUIRE Nominal formation $(G, \tilde{p}, R)$ with Laplacian $L$; New agent $v$ with $\tilde{p}_v$;  each agent $i$ knows $\tilde{p}_i$, $N_i$, $\{L_{ij}\,|\,j\in N_i\}$ and can communicate with its neighbors

\STATE Agent $v$ discovers candidates in sensing range and broadcasts a join request
\STATE Candidates respond with their nominal configurations and neighbor information; $v$ selects two distinct agents $i,j$ from the candidates such that $\{i,j\}\in E$
\STATE Agent $v$ connects to $i$ and $j$, sets $N_i^+\leftarrow N_i\cup\{v\}$, $N_j^+\leftarrow N_j\cup\{v\}$, $N_v^+\leftarrow\{i,j\}$, computes $\Delta_{ij},\Delta_{iv},\Delta_{jv}$ per \eqref{eq:Liju_pad}, and sends $\Delta_{ij}$ to $i,j$, $\Delta_{iv}$ to $i$, $\Delta_{jv}$ to $j$
\STATE Set $L^+_{kl}\leftarrow (\mathcal{E}(L))_{kl}+\Delta_{kl}$ for  $l\in N_k^+$, $k\in\{i,j,v\}$.
\end{algorithmic}
\end{algorithm}

\begin{remark}
Theorem \ref{thm:agent_joining} enables iterative network expansion while preserving Laplacian spectral properties. Unlike centralized optimization approaches \cite{Mukherjee2020, Xiao2022}, Algorithm~\ref{alg:agent_joining} is fully distributed: it requires no global knowledge of the graph topology, agent indexing, or leader/follower roles; each new agent connects to only two existing agents in $\mathbb{R}^d$ via the local matrix-weighted constraint \eqref{eq:iju_matrix}, with the corresponding Laplacian update \eqref{eq:Liju_pad} computed using only relative positions. In contrast, the methods in \cite{Yang2019,Li2025} require each new agent to connect to $d+1$ existing agents with scalar weights and are limited to $\mathbb{R}^2$ and $\mathbb{R}^3$. Moreover, unlike \cite{Li2025} where newly added agents are restricted to be followers, our construction allows the new agent $v$ to be designated either as a follower or as a leader, enabling dynamic leader reallocation—a crucial feature for open multi-agent systems.
\end{remark}

\subsection{Edge Addition} \label{subsec:edge_addition}
We now address Problem~\ref{pbm:topology_update} for edge addition. At first glance, the task appears trivial: simply add a new edge $e =\{j,k \} \notin E$ with an arbitrary positive definite matrix weight. However, such naive addition would destroy the delicate kernel structure $\ker(L)=\varPi$. To characterize admissible perturbations $\Delta = L^+-L$ (with $\bar{\Delta}=0$), define the graph $G_\Delta = (V, E_\Delta)$ where $E_\Delta = \big\{ \{p,q\} \mid \Delta_{pq} \neq 0 \big\}$. Then $G_\Delta$ must satisfy the following necessary condition.

\begin{lemma}\label{lem:minimal_perturbation}
Let $L\in \mathbb{R}^{dn \times dn}$ be a matrix-valued Laplacian with $\ker(L)=\varPi$ and let $\Delta$ be a symmetric matrix such that $L^+=L+\Delta$ also satisfies $\ker(L^+)=\varPi$. 
Then in $G_\Delta$, the degree of every non‑isolated vertex is at least $2$.
\end{lemma}
\begin{pf}
Since $\ker(L)=\ker(L^+)=\varPi$, for any $p\in\varPi$ we have $Lp=0$ and $L^+p=0$, hence $\Delta p=0$. Thus $\varPi\subseteq\ker(\Delta)$. Choosing $p = \mathbf{1}_n\otimes\tau$ gives $\left( \sum_{u \in V} \Delta_{ju} \right) \tau = 0$ for all $\tau$, so $\sum_{u \in V} \Delta_{ju} = 0$ for each vertex $j$.

Suppose there exists a vertex $j$ that has exactly one neighbor $k$ in the graph induced by the nonzero blocks of $\Delta$, i.e., $\Delta_{jk}\neq0$ and $\Delta_{jw}=0$ for all $w\neq j,k$. From the row‑sum condition, $\Delta_{jj}+\Delta_{jk}=0$. For any $p\in\varPi$, the $j$-th block of $\Delta p=0$ reads
\begin{equation}
\Delta_{jj}p_j + \Delta_{jk}p_k = \Delta_{jj}(p_j-p_k)=0.
\end{equation}
Since this must hold for all $p\in\varPi$ and under Assumption~\ref{ass:graph}, we conclude $\Delta_{jj}=\Delta_{jk}=0$, a contradiction. Hence no vertex can have degree exactly $1$. If a vertex has no neighbors, its entire row of $\Delta$ is zero, which is allowed. Therefore, every non‑isolated vertex has degree at least $2$.
\end{pf}

By Lemma~\ref{lem:minimal_perturbation}, any nontrivial perturbation must involve at least three vertices; the simplest such structure is a triangle, whose associated constraint is exactly \eqref{equ_vvc}. To add only the single edge $e$, we therefore construct a perturbation $\Delta$ that corresponds to the Laplacian of a cycle containing $e$, which can be built from triangle primitives.

\subsubsection{Construct the perturbation as the Laplacian of a cycle via triangulation}
By $2$-vertex-connectivity, there exists a path from $j$ to $k$. Without loss of generality, let its vertex sequence be $1,2,\dots,q$ with $1=j$ and $q=k$. Adding $e=\{j,k \}$ creates the cycle $C = (1,2,\dots,q,1)$. To construct a $\Delta$ that adds exactly $e$, we adopt a star-shaped triangulation of $C$ with $1$ (i.e., $j$) as the common vertex. Specifically, we add edges $\{ \{1,i\} \,|\, i=3,\dots,q-1 \}$, yielding triangles $\triangle_{1,2,3},\triangle_{1,3,4},\dots,\triangle_{1,q-1,q}$. For each triangle $\triangle_{1,i,i+1}$ ($i=2,\dots,q-1$), construct the $3d\times 3d$ local matrix $L^{1i(i+1)}(D_i)$ as in \eqref{eq:ijv_block} with the vertex order $(1,i,i+1)$. Then embed it into the $dn\times dn$ matrix $L^{1i(i+1)}_{\text{pad}}(D_i)=0$. Specifically, let $(a_1,a_2,a_3) = (1,i,i+1)$. For $b,d \in\{1,2,3\}$, 
\begin{equation}\label{eq_pad}
    \bigl(L^{1i(i+1)}_{\text{pad}}(D_i)\bigr)_{a_ba_d} = \bigl(L^{1i(i+1)}(D_i)\bigr)_{bd}.
\end{equation}
The total perturbation is
\begin{equation}\label{eq:cycleL}
\Delta = \sum_{i=2}^{q-1} L^{1i(i+1)}_{\text{pad}}(D_i).
\end{equation}
We now cancel the contributions of the temporary diagonals. For each interior diagonal $\{1,i\}$ with $i=3,\dots,q-1$, which is shared by triangles $\triangle_{1,i-1,i}$ and $\triangle_{1,i,i+1}$, cancellation of its contributions requires
\begin{equation}
\label{eq:cancel}
W_{i,i-1}^\top D_{i-1}W_{i-1,1} + W_{i,i+1}^\top D_i W_{i+1,1}=0.
\end{equation}

Substituting $W_{ab} = \diag(\tilde{p}_{ab,R}) R^{\top}$ into \eqref{eq:cancel} and simplifying with $R^\top R = I$ gives a diagonal condition. For each coordinate $l=1,\dots,d$, let $d_i^l$ denote the $l$-th diagonal entry of $D_i$, we obtain the scalar recurrence
\begin{equation}
\label{eq:recurrence}
d_i^l = \gamma_{i-1}^l \, d_{i-1}^l, \, \gamma_{i-1}^l = \frac{\tilde p_{i-1,i,R}^l  \tilde p_{i-1,1,R}^l}{\tilde p_{i,i+1,R}^l  \tilde p_{i+1,1,R}^l}, \, i=3,\dots,q-1.
\end{equation}

Under Assumption~\ref{ass:graph}, $\tilde p_{ab,R}^l \neq 0$ for all distinct $a,b\in V$ and all $l=1,\dots,d$, so the geometric factors $\gamma_i^l$ are well-defined and non-zero. These factors completely determine the feasibility of the construction. Once $D_2\succ0$ is chosen, all subsequent $D_i$ are uniquely determined by the recurrence \eqref{eq:recurrence}, and the necessary and sufficient condition for $D_i\succ0$ is $\gamma_i^l>0$ for all $i,l$. This leads to the following characterization.

\begin{theorem}\label{thm:edge_addition}
Let $F=(G,\tilde{p},R)$ satisfy Assumption~\ref{ass:graph} and $L$ be its Laplacian with a formation spectrum. For a new edge $e=\{j,k\}\notin E$, suppose that there exists a path with vertex sequence $j=1,2,\dots,q=k$ such that all geometric factors $\gamma_i^l$ defined in \eqref{eq:recurrence} are positive. Then, by constructing $\Delta$ via \eqref{eq:cycleL} with $D_i$ specified in \eqref{eq:recurrence} where $D_2$ is designed to be positive definite, there exists $\varepsilon>0$ such that the perturbed Laplacian $L^+=L+\varepsilon\Delta$ has a formation spectrum and $E_{\text{add}}=\{e\}$ is minimal.
\end{theorem}

\begin{pf}
By $2$-vertex-connectivity, the path with vertex sequence $j=1,2,\dots,q=k$ exists. Triangulating the cycle $C=(1,2,\dots,q,1)$ yields a set of triangles. Each edge of the cycle receives a nonzero matrix weight contribution from a single triangle, while each diagonal receives contributions from two triangles and is canceled by condition \eqref{eq:cancel}. Consequently, the nonzero off-diagonal blocks $\Delta_{ij}$ of $\Delta$ are precisely the edge weights of the cycle $C$. By further scaling all $D_i$ with a sufficiently small common factor $\varepsilon>0$, we ensure that the contributions to the original edges remain small enough to keep them nonzero, while the new edge $e$ receives a nonzero contribution; hence only the edge $e$ is added to $G$. Since all geometric factors defined in \eqref{eq:recurrence} are positive, then $\Delta \succeq 0$; moreover, any principal submatrix $\Delta_{ff}$ is also positive semidefinite. Since $L_{ff}\succ0$ and $\Delta_{ff} \succeq 0$, we obtain $L^+_{ff}=L_{ff}+\varepsilon\Delta_{ff}\succ0$. The remainder of the proof that $L^+$ has a formation spectrum follows the same structure as Theorem~\ref{thm:agent_joining}.
\end{pf}

\subsubsection{Construction feasibility independent of triangulation}
One might worry that the above analysis depends on the particular triangulation chosen. Since the number of triangulations of a cycle with \(q\) vertices is the \((q-2)\)-th Catalan number \cite{Loera2010}, which grows exponentially with \(q\), searching over all possible triangulations for one admitting positive definite design matrices is computationally prohibitive. The following result shows that such a search is unnecessary: the feasibility condition is intrinsic to the cycle itself and independent of how it is triangulated.


\begin{theorem}
\label{thm:triangulation_equiv}
Let $C = (1,2,\dots,q,1)$ be a cycle with nominal configuration $(\tilde p,R)$ satisfying Assumption~\ref{ass:graph}. Then either every triangulation of $C$ admits positive definite diagonal matrices $\{D_i\}_{i=2}^{q-1}$ (as used in the construction $\Delta$ defined in \eqref{eq:cycleL}) satisfying \eqref{eq:cancel}, or none does. In other words, the existence of such matrices is independent of the chosen triangulation.
\end{theorem}

\begin{pf}
Any two triangulations are connected by a finite sequence of edge flips \cite{Loera2010}. It suffices to show that a single flip preserves existence.

Consider a quadrilateral with vertices $a,b,c,d$ in cyclic order. Two triangulations are possible: $T_1$ with diagonal $(a,c)$ (triangles $\triangle_{abc}$ and $\triangle_{acd}$) and $T_2$ with diagonal $(b,d)$ (triangles $\triangle_{bcd}$ and $\triangle_{bda}$). Note that the local matrix $L^{xyz}=M_{xyz}^\top D_{xyz} M_{xyz}$ for a triangle is invariant under vertex permutation. Write $d_{abc}^l$, $d_{acd}^l$, $d_{bcd}^l$, $d_{bda}^l$ for the $l$-th diagonal entries of the design matrices $D$ associated with each triangle. Then \eqref{eq:cancel} yields:
\begin{equation}
d_{acd}^l = \gamma_{ac}^l d_{abc}^l,\quad 
d_{bda}^l = \gamma_{bd}^l d_{bcd}^l,
\end{equation}
where
\begin{equation}
\gamma_{ac}^l = \frac{\tilde p_{bc,R}^l \tilde p_{ba,R}^l}{\tilde p_{cd,R}^l \tilde p_{da,R}^l},\qquad 
\gamma_{bd}^l = \frac{\tilde p_{cd,R}^l \tilde p_{bc,R}^l}{\tilde p_{da,R}^l \tilde p_{ba,R}^l}.
\end{equation}
Assumption~\ref{ass:graph} guarantees that the geometric factors $\gamma_{ac}^l$ and $\gamma_{bd}^l$ are well-defined and non-zero. A direct computation shows $\gamma_{ac}^l \gamma_{bd}^l = (\tilde p_{bc,R}^l / \tilde p_{da,R}^l)^2 > 0$ and $\gamma_{ac}^l / \gamma_{bd}^l > 0$, so $\gamma_{ac}^l$ and $\gamma_{bd}^l$ have the same sign. If $T_1$ admits positive definite solutions, then $d_{abc}^l>0$, $d_{acd}^l>0$, and $\gamma_{ac}^l>0$. Hence $\gamma_{bd}^l>0$. Choosing any $d_{bcd}^l>0$ and setting $d_{bda}^l = \gamma_{bd}^l d_{bcd}^l$ yields positive entries, giving positive definite solutions for $T_2$. The converse follows by symmetry. 
\end{pf}

\begin{algorithm}[htbp]
	\caption{Distributed Edge Addition Protocol}
	\label{alg:edge_addition}
	\begin{algorithmic}[1]
		\REQUIRE $(G, \tilde{p}, R)$, $L$; agents $j,k$ with edge $\{j,k\}$ to be added.
		
		\STATE \textbf{Path discovery:} Agent $k$ establishes communication with $j$, exchange their nominal positions $\tilde{p}_j$, $\tilde{p}_k$, and update their neighbor sets: $N_j = N_j \cup \{k\}$, $N_k = N_k \cup \{j\}$.
		
		\STATE Agent $k$ initiates a probe packet containing $\mathbb{P} = [j,k]$, $\tilde{\mathbb{P}} = [\tilde{p}_j, \tilde{p}_k]$, $\Gamma = [~]$ and sends it to all neighbors $v \in N_k \setminus \{j\}$.
		
		\WHILE{probe not reached $j$}
			\STATE At current agent $v$ with packet ($\mathbb{P}=[j,\dots, f]$, $\tilde{\mathbb{P}}=[\tilde{p}_j,\dots, \tilde{p}_f]$, $\Gamma$):
			\FOR{each $w \in N_v$}
				\IF{$w = j$}
					\STATE Forward probe directly to $j$ with packet $[\mathbb{P}, v]$, $[\tilde{\mathbb{P}}, \tilde{p}_v]$, and $\Gamma$.
				\ELSIF{$w \notin \mathbb{P}$}
					\STATE For each coordinate $l=1,\dots,d$, compute the geometric factor
					\[
					\gamma^l = \frac{\tilde p_{f,v,R}^l \tilde p_{f,j,R}^l}{\tilde p_{v,w,R}^l \tilde p_{w,j,R}^l}.
					\]
					\IF{$\gamma^l > 0$ for all $l$}
					\STATE Forward probe to $w$ with updated packet $[\mathbb{P}, v]$, $[\tilde{\mathbb{P}}, \tilde{p}_v]$, and $[\Gamma, \gamma]$.
					\ENDIF
				\ENDIF
			\ENDFOR
		\ENDWHILE
		\STATE \textbf{$\Delta$ computed by $j$:} The probe reaches $j$ carrying the complete cycle information in $\mathbb{P}$. Using $\tilde{\mathbb{P}}$ and $\Gamma$, $j$ computes $\Delta$ via \eqref{eq:cycleL}, where $D_i$, $i\in\{2,...,q-1\}$ are constructed from arbitrary $D_2\succ0$ and $\Gamma$ via \eqref{eq:recurrence}. If multiple feasible paths exist, just use the first one.
		\STATE \textbf{Edge weight distribution:} $j$ sends an update packet along $\mathbb{P}$ containing the computed edge weights for all edges in the cycle $C$. Each involved agent updates its local copy of the weights for its incident edges in $C$.
	\end{algorithmic}
\end{algorithm}

\subsubsection{Finding a feasible cycle}
Theorem~\ref{thm:triangulation_equiv} liberates us from the exponential tyranny of triangulation choices. Yet a new challenge emerges: in large-scale networks, the edge $e$ may belong to exponentially many cycles, each corresponding to a distinct path from $j$ to $k$. Different cycles yield different geometric factors, and only those with all $\gamma_i^l>0$ are feasible. Thus, even after eliminating the search over triangulations, we still face an exponential search over paths—a problem that seems equally intractable.

Crucially, while any triangulation suffices for the final construction, only those with a local structure can guide the search itself. The star-shaped triangulation with a common vertex possesses exactly this property: each geometric factor $\gamma_i^l$ depends only on four vertices—the current vertex, its predecessor, the candidate successor, and the common vertex (see \eqref{eq:recurrence}). This observation leads to a dynamic pruning strategy, implemented in Algorithm~\ref{alg:edge_addition}: at each step, any extension that violates $\gamma_i^l>0$ is discarded. Because the condition is local, pruning can be performed in a fully distributed manner. The pruning is both sound—any feasible path satisfies the condition at every step and therefore survives—and effective, eliminating large numbers of infeasible paths early.

When the geometric factors are not all positive, the single-edge addition via cycle triangulation fails. In this case, one can add a complete triangle $\triangle_{j,k,v}$ with a neighbor $v$ of $j$ or $k$; set $\Delta$ to the perturbation corresponding to $\triangle_{j,k,v}$. This adds two edges simultaneously and always preserves the spectral properties. 

\begin{remark}
Unlike \cite{Fathian2019}, which only achieves regular polygon formations under cyclic sensing, our construction realizes both regular and non‑regular polygon formations in a unified way. The key is to build a Laplacian for a cycle via triangulation. This construction requires a cycle $C$ whose geometric factors in \eqref{eq:recurrence} are all positive. Once such a cycle is found, any triangulation of $C$ succeeds—a fact established by Theorem~\ref{thm:triangulation_equiv}. The remaining combinatorial challenge of finding a feasible cycle reduces to a path search problem, which can be solved efficiently by our distributed algorithm via pruning the search space using local positivity checks enabled by the locality of the geometric factors. Moreover, unlike existing methods tailored to minimally rigid graphs \cite{Jing2019, Cao2020, Chen2021, Farid2024}, our edge-addition strategy allows the formation to evolve seamlessly from any $2$-vertex-connected graph up to the complete graph while preserving the required spectral properties.
\end{remark}

\subsection{Agent Leaving}
\label{sol:agent_leaving}
We now address the agent leaving operation in Problem~\ref{pbm:topology_update}. Consider vertex $u$ to be removed from the $2$-vertex-connected graph $G$, with $|N_u|\in[2, |V|-1]$. Order all the vertices as $U = V \setminus (N_u \cup \{u\})$, $N = N_u$, and $u$. In this ordering, the Laplacian $L$ admits the structure
\begin{equation}
L = \begin{bmatrix}
L_{UU} & L_{UN} & 0 \\
L_{NU} & L_{NN} & L_{Nu} \\
0 & L_{uN} & L_{uu}
\end{bmatrix},
\label{eq:L_block_leaving}
\end{equation}
where $L_{uu} \succ 0$. Since $\ker(L) = \varPi(\tilde{p},R)$, we have $\varPi(\tilde{p},R) \subseteq \ker(L_u)$ with $L_u = \begin{bmatrix} 0 & L_{uN} & L_{uu} \end{bmatrix} \in \mathbb{R}^{d \times dn}$ denoting the block row corresponding to $u$.

A naive approach would simply delete $u$ and all its incident edges from $G$ and zero out the corresponding blocks in $L$. Such naive removal, however, may cause that the remaining graph loses $2$-vertex-connectivity, and the resulting Laplacian loses the formation spectrum. To overcome this, we follow the update formula \eqref{eq:L_update} and set $\bar{\Delta} = 0$, while $\Delta$ is chosen as
\begin{equation}
\Delta = -L_u^{\top} L_{uu}^{-1} L_u.
\label{eq:deltau}
\end{equation}
The Laplacian of the updated formation $(G^+,\tilde{p}^+,R)$ is
\begin{equation}
L^+ = \mathcal{E}\bigl(L - L_u^{\top} L_{uu}^{-1} L_u\bigr).
\label{eq:LVL}
\end{equation}

\begin{theorem} \label{thm:agent_leaving}
Let $F=(G,\tilde{p},R)$ satisfy Assumption~\ref{ass:graph} and $L$ be a Laplacian with a formation spectrum. Removing agent $u \in V$ and all edges incident to $u$ using $\bar{\Delta}=0$ and $\Delta$ as in \eqref{eq:deltau} yields a formation $F^+ = (G^+,\tilde{p}^+,R)$ whose Laplacian $L^+$ given by \eqref{eq:LVL} has a formation spectrum.
\end{theorem}
\begin{pf}
We first show $L^+ \succeq 0$. Write $L$ in block form as
\begin{equation}
L = \begin{pmatrix} A & B \\ B^\top & C \end{pmatrix}, \quad
L_u = \begin{pmatrix} B^\top & C \end{pmatrix},
\end{equation}
where $C = L_{uu}$ is symmetric positive definite and $A$ corresponds to the remaining agents. Then
\begin{equation}
L_u^{\top} L_{uu}^{-1} L_u = \begin{pmatrix} B C^{-1} B^\top & B \\ B^\top & C \end{pmatrix},
\end{equation}
so
\begin{equation}
M = L - L_u^{\top} L_{uu}^{-1} L_u = \begin{pmatrix} A - B C^{-1} B^\top & 0 \\ 0 & 0 \end{pmatrix}.
\end{equation}
Since $L \succeq 0$ and $C \succ 0$, the Schur complement $S = A - B C^{-1} B^\top \succeq 0$. Hence $L^+ = \mathcal{E}(M) = S \succeq 0$. Moreover, by the Schur complement rank theorem, $\operatorname{rank}(S) = \operatorname{rank}(L) - \operatorname{rank}(C) = (n-1)d - 2d$.

We next prove $Q^{\top} L^+ Q \succ 0$ for all $Q \in \mathcal{S}$. For any $Q \in \mathcal{Q}$ selecting two leaders with indices $\mathcal{L}$, let $\mathcal{F}$ include the follower indices. Partition $A$ and $B$ accordingly:
\begin{equation}
A = \begin{pmatrix} A_{\mathcal{F}\mathcal{F}} & A_{\mathcal{F}\mathcal{L}} \\ A_{\mathcal{L} \mathcal{F}} & A_{\mathcal{L}\mathcal{L}} \end{pmatrix}, \quad
B = \begin{pmatrix} B_{\mathcal{F}} \\ B_{\mathcal{L}} \end{pmatrix}.
\end{equation}
Then $Q^{\top} L^+ Q = A_{\mathcal{F}\mathcal{F}} - B_{\mathcal{F}} C^{-1} B_{\mathcal{F}}^{\top}$. In the original formation, the follower block for leaders $\mathcal{L}$ is
\begin{equation}
L_{ff} = \begin{pmatrix}
A_{\mathcal{F}\mathcal{F}} & B_{\mathcal{F}} \\
B_{\mathcal{F}}^{\top} & C
\end{pmatrix},
\end{equation}
where $L_{ff}\succ 0$ because $L$ has a formation spectrum. The Schur complement of $C$ in $L_{ff}$ is $A_{\mathcal{F}\mathcal{F}} - B_\mathcal{F} C^{-1} B_\mathcal{F}^{\top}$, which is therefore positive definite. Thus $Q^{\top} L^+ Q \succ 0$.

We now show $\ker(L^+) = \varPi(\tilde{p}^+,R)$. For any $p \in \varPi(\tilde{p},R)$, $L_u p = 0$ implies $(L - L_u^{\top} L_{uu}^{-1} L_u)p = Lp = 0$, so $\varPi(\tilde{p},R) \subseteq \ker(M)$. Removing $u$ gives $\varPi(\tilde{p}^+,R) \subseteq \ker(L^+)$. Since $\operatorname{rank}(L^+) = (n-1)d - 2d$, we have $\dim(\ker(L^+)) = 2d = \dim(\varPi(\tilde{p}^+,R))$, hence $\ker(L^+) = \varPi(\tilde{p}^+,R)$.

Thus $L^+$ has a formation spectrum. By Lemma \ref{lem:2vc_sp}, $G^+$ is $2$-vertex-connected. 
\end{pf}

The number of added edges satisfies $0 \le |E_{\text{add}}| \le \binom{|N_u|}{2}$, where the upper bound corresponds to forming a complete graph among the neighbors of the departing agent. The distributed Laplacian update protocol is similar to Algorithm~\ref{alg:agent_joining}, hence omitted.

\begin{remark}
Compared with existing works on graph biconnectivity maintenance \cite{Eswaran1976, Rosenthal1977, Ahmadi2006}, which focus exclusively on topological connectivity, our algorithm preserves the critical spectral properties of the Laplacian while ensuring that the resulting graph remains $2$‑vertex‑connected. Unlike the affine framework construction in \cite{Li2025}, which requires a predefined parent–child hierarchical structure, our method operates purely on local neighborhood information. Moreover, our approach imposes no restriction on the number of neighbors of the departing agent.
\end{remark}

\subsection{Edge Removal}
\label{subsec:edge_removal}
We now address the edge removal operation in Problem~\ref{pbm:topology_update}. By lemma \ref{lem:2vc_sp} and the fact that Laplacian $L$ has a formation spectrum, $G$ is $2$-vertex-connected. Removing $e=\{j,k\}\in E$ yields $H=G-e$, which may not satisfy the $2$-vertex-connectivity required for $G^{+}$. The following theorem characterizes the minimal number of compensatory edges required to restore $2$-vertex-connectivity.

\begin{theorem}\label{thm:main-edge-removal}
Let $G=(V,E)$ be a $2$-vertex-connected graph with $|V|\ge 4$, $e=\{j,k\}\in E$, and $H=G-e$ with $d_H(v)$ denoting the degree of $v$ in $H$. Then $H$ is $2$-vertex-connected if and only if there exist two cycles $C_1,C_2$ in $G$ such that $e\in E(C_1)\cap E(C_2)$ and $V(C_1)\cap V(C_2)=\{j,k\}$. Otherwise, the minimum number of new edges (distinct from $e$) needed to restore $2$-vertex-connectivity is 1 when $\max\{d_H(j),d_H(k)\}\ge 2$, and 2 when $d_H(j)=d_H(k)=1$.
\end{theorem}

\begin{pf}
We prove it by considering three cases.

\textbf{Case (i): $H$ is $2$-vertex-connected.}
    \textit{(Necessity)} If $H$ is $2$-vertex-connected, then by Menger's theorem there exist two internally vertex-disjoint paths $P_1,P_2$ with endpoints $j$ and $k$. Adding the edge \(e\) to \(P_1\) and \(P_2\) yields two cycles \(C_1\) and \(C_2\), both containing \(e\) and satisfying $V(C_1) \cap V(C_2) = \{j,k\}$.

    \textit{(Sufficiency)} Conversely, suppose there are two cycles $C_1,C_2$ with $e\in E(C_1)\cap E(C_2)$ and $V(C_1)\cap V(C_2)=\{j,k\}$. Let $P_i=C_i-e$ ($i=1,2$); then $P_1,P_2$ are internally vertex-disjoint paths. We prove that $G-e$ is $2$-vertex-connected by showing that $G-e-x$ is connected for every $x \in V$. 
    
    \emph{Subcase 1: $x = j$ (symmetric for $x = k$).} 
    Then $G-e-j = G-j$. Since $G$ is $2$-vertex-connected, $G-e-j$ is connected.

    \emph{Subcase 2: $x \notin \{j, k\}$.} 
    The graph $G-x$ is connected because $G$ is $2$-vertex-connected. If $G-x-e$ is disconnected then $G-x$ contains no path between $j$ and $k$ other than the edge $e$. However, since $x \notin \{j, k\}$ and $P_1, P_2$ are internally vertex-disjoint, at most one of them contains $x$. Hence at least one of $P_1, P_2$ remains a path between $j$ and $k$ in $G-x$, a contradiction. Thus $(G-e)-x$ is connected.


\textbf{Case (ii): $H$ is not $2$-vertex-connected and $\max\{d_H(j)$ $,d_H(k)\}\ge 2$.}
Without loss of generality, let $d_H(j)\ge 2$. We first show that there exists a vertex $w\in N_j(H)\setminus N_k(H)$ which is not a cut vertex of $H$. Suppose, for contradiction, that no such $w$ exists. Let $u,v\in N_j(H)$ be distinct. The following statements hold.
    \begin{itemize}
    \item If $u,v\in N_k(H)$, then the cycles $(j,u,k,j)$ and $(j,v,k,j)$ satisfy the condition of Case (i), contradiction.
    
    \item If $u$ is a cut vertex but $v\in N_k(H)$, then $H-u$ separates $j$ and $k$ (since $H-u = G-u-e$ is disconnected and $G-u$ is connected, $H-u$ has exactly two components separated by $e$). However, the path from $j$ to $v$ to $k$ connects them in $H-u$, a contradiction. The symmetric case is similar.

    \item If both $u$ and $v$ are cut vertices of $H$, then $H-u$ and $H-v$ separates $j$ and $k$ (similar to the proof above). Because $G$ is $2$-vertex-connected, there is a path $P$ from $j$ to $k$ in $G$ avoiding $e$. Then $P$ lies in $H$ and must pass through both $u$ and $v$ (otherwise it contradicts that $H-u, H-v$ separate $j,k$). Then $P$ contains a subpath from $v$ to $k$ avoiding $u$, so in $H-u$ the edge $\{j,v\}$ together with this subpath connects $j$ to $k$, contradiction.
    \end{itemize}
    Thus, such a $w$ exists. Now add edge \(\{k, w\}\) to \(H\), yielding \(H'\). We show $H'$ is $2$-vertex-connected. Take any $x\in V$,
    \begin{itemize}
        \item if $x\in\{j,k\}$, then $H'-x$ contains $G-x$ (connected);
        \item if $x\notin\{j,k\}$ and $x$ is not a cut vertex of $H$, then $H-x$ connected implies $H'-x$ connected;
        \item if $x\notin\{j,k\}$ and $x$ is a cut vertex of $H$, then $H-x$ separates $j$ and $k$ (similar to the proof above). Since $w$ is a non‑cut vertex, $w\neq x$, and $w \in N_j(H)$ lies in the component of $j$. The new edge $f=\{k,w\}$ thus connects the two components, so $H'-x$ is connected.
    \end{itemize}

\textbf{Case (iii): $H$ is not $2$-vertex-connected and $d_H(j)=d_H(k)=1$.} Let $v$ and $w$ be the unique neighbours of $j$ and $k$ in $H$, respectively. We first show $v\neq w$. If $v=w$, then $G$ contains the triangle $\triangle_{j,k,v}$. Since $|V|\ge 4$, there exists a vertex $x\neq j,k,v$. Because $G$ is $2$-vertex-connected, $G-v$ is connected, so $x$ must be adjacent to $j$ or $k$ in $G$. But then $d_H(j)\ge 2$ or $d_H(k)\ge 2$, contradicting $d_H(j)=d_H(k)=1$. Thus $v\neq w$. Adding a single edge to $H$ can eliminate at most one of the leaves $j,k$, leaving the other leaf’s neighbor as a cut vertex. Hence adding edges $f_1=\{j,w\}$ and $f_2=\{k,v\}$ to \(H\) yields \(H'\), which is  $2$-vertex-connected; the verification is analogous to Case~(ii) and omitted.
\end{pf}

\begin{figure}[htbp]
\centering
\begin{subfigure}[b]{0.32\linewidth}
    \centering
    \begin{tikzpicture}[scale=0.6, every node/.style={circle, fill=black, inner sep=1pt, minimum size=2pt}]
    \node (j) at (3,0) [label=right:$j$] {};
    \node (k) at (2,1.5) [label=above:$k$] {};
    \node (a) at (1,1.5) [label=above:$a$] {};
    \node (b) at (0,0) [label=left:$b$] {};
    \node (c) at (1,-1.5) [label=below:$c$] {};
    \node (w) at (2,-1.5) [label=below:$w$] {};
    
    \draw[red, dashed, thick] (j) -- (k);
    
    \draw (k) -- (a) -- (b) -- (j);
    \draw (b) -- (c) -- (w) -- (j);
    
    \draw[green,dashed, thick] (k) -- (w) node[midway, above left, draw=none, fill=none] {f};

    \end{tikzpicture}
    \caption*{(a)}
\end{subfigure}
\hfill
\begin{subfigure}[b]{0.32\linewidth}
    \centering
    \begin{tikzpicture}[scale=0.6, every node/.style={circle, fill=black, inner sep=1pt, minimum size=2pt}]
    \node (j) at (2,1.5) [label=above:$j$] {};
    \node (k) at (3,0) [label=right:$k$] {};
    \node (v) at (1,1.5) [label=above:$v$] {};
    \node (b) at (0,0) [label=left:$b$] {};
    \node (c) at (1,-1.5) [label=below:$c$] {};
    \node (w) at (2,-1.5) [label=below:$w$] {};
    
    \draw[red, dashed, thick] (j) -- (k);
    
    \draw (j) -- (v) -- (b) -- (c) -- (w) -- (k);
    
    \draw[green,dashed, thick] (j) -- (w) node[midway, below left, draw=none, fill=none] {$f_1$};
    \draw[green,dashed, thick] (k) -- (v) node[midway, left, draw=none, fill=none] {$f_2$};
    
    \end{tikzpicture}
    \caption*{(b)}
\end{subfigure}
\hfill
\begin{subfigure}[b]{0.32\linewidth}
    \centering
    \begin{tikzpicture}[scale=0.6, every node/.style={circle, fill=black, inner sep=1pt, minimum size=2pt}]
    \node (j) at (0,0) [label=left:$j$] {};
    \node (k) at (3,0) [label=right:$k$] {};
    \node (a) at (1,1.5) [label=above:$a$] {};
    \node (b) at (2,1.5) [label=above:$b$] {};
    \node (c) at (1,-1.5) [label=below:$c$] {};
    \node (d) at (2,-1.5) [label=below:$d$] {};
    
    \draw[red, dashed, thick] (j) -- (k);
    
    \draw (j) -- (a) -- (b) -- (k);
    \draw (j) -- (c) -- (d) -- (k);
    \end{tikzpicture}
    \caption*{(c)}
\end{subfigure}

\caption{Three cases of restoring $2$-vertex connectivity after removing edge $\{j,k\}$ from $G$. (a) Adding one compensatory edge $\{k,w\}$ suffices to restore $2$-vertex connectivity. (b) Two compensatory edges $\{j,w\}$ and $\{k,v\}$ are required. (c) No compensatory edge is needed.}
\label{fig:edge_removal_cases}
\end{figure}
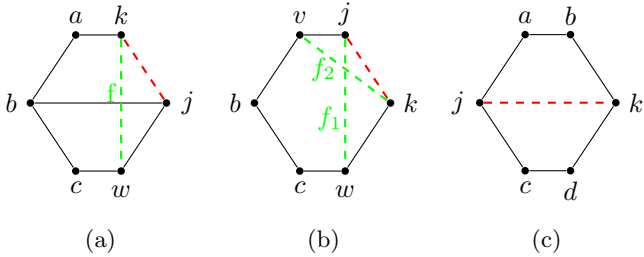

For edge removal, we take $\mathcal{E}$ as the identity and $\bar{\Delta} = 0$, so the update reduces to $L^+ = L + \Delta$. We now translate each case in Theorem~\ref{thm:main-edge-removal} into a concrete construction of $\Delta \succeq 0$, presenting them in the order (ii), (iii), (i) for conciseness.

\textbf{Case (ii): $H$ is not $2$-vertex-connected and $\max\{d_H(j)$ $,d_H(k)\}\ge 2$.} A single compensatory edge suffices to restore $2$-vertex connectivity (see Fig.~\ref{fig:edge_removal_cases}(a)). Without loss of generality, assume $d_H(j) \ge 2$ and let $w \in N_j(H) \setminus N_k(H)$ be a non-cut vertex of $H$ (its existence is guaranteed in the proof of Theorem~\ref{thm:main-edge-removal} and such $w$ can be determined locally via two-hop neighbor information of agent $j$ \cite{Wang2026}). For each candidate $w$, seek $D \succ 0$ such that the perturbation $L^{jkw}_{\text{pad}}(D)$ (for $\triangle_{j,k,w}$ as in \eqref{eq_pad}) satisfies:
\begin{equation}
L_{kj}+W_{wj}^\top D W_{kw}=0, \quad L_{jw} + W_{kw}^\top D W_{jk} \neq 0.
\end{equation}
The first condition cancels edge $\{j,k\}$, while the second preserves edge $\{j,w\}$. If such $D \succ 0$ exists, we set 
\begin{equation}
\Delta = L^{jkw}_{\text{pad}}(D),
\end{equation}
completing the edge removal operation for this case. Here $E_{\text{add}} = \{\{k,w\}\}$ and $|E_{\text{add}}|=1$ is minimal by Theorem~\ref{thm:main-edge-removal}. The corresponding distributed Laplacian update protocol is similar to Algorithm \ref{alg:agent_joining}, and is therefore omitted.

\textbf{Case (iii): $H$ is not $2$-vertex-connected and $d_H(j)=d_H(k)=1$.} Let $v$ and $w$ be the unique neighbors of $j$ and $k$ in $H$, with $v \neq w$ (see Fig.~\ref{fig:edge_removal_cases}(b)). Find $D_1 \succ 0$ and $D_2 \succ 0$ such that the perturbations $L^{jkw}_{\text{pad}}(D_1)$ (for $\triangle_{j,k,w}$ as in \eqref{eq_pad}) and $L^{jkv}_{\text{pad}}(D_2)$ (for $\triangle_{j,k,v}$ as in \eqref{eq_pad}) satisfy:
\begin{equation}
\begin{split}
&L_{jk} + W_{kw}^\top D_1 W_{wj} + W_{kv}^\top D_2 W_{vj} = 0, \\ 
&L_{jv} + W_{kv}^\top D_2 W_{jk} \neq 0, \quad L_{kw} + W_{wj}^\top D_1 W_{jk} \neq 0.
\end{split}
\end{equation}
The first condition cancels edge $\{j,k\}$, while the second and third preserve edges $\{j,v\}$ and $\{k,w\}$, respectively. If such $D_1, D_2 \succ 0$ exist, we set 
\begin{equation}
\Delta = L^{jkw}_{\text{pad}}(D_1) + L^{jkv}_{\text{pad}}(D_2),
\end{equation}
completing the edge removal operation for this case. Here $E_{\text{add}} = \{\{j,w\}, \{k,v\}\}$ and $|E_{\text{add}}|=2$ is minimal by Theorem~\ref{thm:main-edge-removal}. The corresponding distributed Laplacian update protocol is similar to Algorithm \ref{alg:agent_joining}, and is therefore omitted.

\textbf{Case (i): $H$ is $2$-vertex-connected (see Fig.~\ref{fig:edge_removal_cases}(c)).} Since $H$ is already $2$-vertex-connected, no compensatory edges are needed. However, the Laplacian $L$ still contains the contribution of edge $\{j,k\}$, which must be algebraically eliminated. We achieve this by considering an inverse operation of the edge addition method presented in Section~\ref{subsec:edge_addition}. Specifically, we construct a cycle-based perturbation matrix that exactly removes $\{j,k\}$ without introducing any new edges. The key modifications to Algorithm~\ref{alg:edge_addition} are as follows.

At agent $k$, for each $v \in N_k \setminus \{j\}$, we seek a design matrix $D_k \succ 0$ such that the perturbation $L^{jkv}_{\text{pad}}(D_k)$ (constructed for triangle $\triangle_{j,k,v}$ as in \eqref{eq_pad}) satisfies:
\begin{equation}
L_{kj}+W_{vj}^\top D_k W_{kv}=0, \quad L_{kv}+W_{vj}^\top D_k W_{jk} \ne 0.
\end{equation}
The first condition cancels edge $\{j,k\}$, while the second preserves edge $\{v,k\}$. The perturbation also introduces a nonzero contribution to edge $\{j,v\}$, which will be canceled in a later step. If such $D_k \succ 0$ exists, agent $k$ sends a probe to $v$ containing $\mathbb{D}=[D_k]$ and the information in Algorithm~\ref{alg:edge_addition}. 

At current agent $v$ with packet ($\mathbb{P}=[j,\dots, f]$, $\tilde{\mathbb{P}}=[\tilde{p}_j,\dots, \tilde{p}_f]$, $\mathbb{D}=[D_k, \dots, D_f]$, $\Gamma$). For each $w \in N_v$:
\begin{itemize}
\item If $w = j$, the probe is forwarded directly to $j$ with the current packet after verifying that 
\begin{equation}
    L_{jv} + W_{fv}^\top D_f W_{jf} \neq 0.
\end{equation}
\item Else if $w \notin \mathbb{P}$, we seek a design matrix $D_v \succ 0$ such that the perturbation $ L^{jvw}_{\text{pad}}(D_v) $ (constructed for triangle $\triangle_{j,v,w}$ as in \eqref{eq_pad}) satisfies:
\begin{equation}
\begin{split}
    & W_{v f}^\top D_f W_{f j} + W_{vw}^\top D_v W_{w j} = 0, \\
    &L_{vw} + W_{wj}^\top D_v W_{jv} \neq 0.
\end{split}
\end{equation}
The first condition cancels the contribution to edge $\{j,v\}$ introduced by the previous perturbation $L^{jfv}_{\text{pad}}(D_f)$, while the second preserves edge $\{v,w\}$. The perturbation $L^{jvw}_{\text{pad}}(D_v)$ also introduces a nonzero contribution to edge $\{j,w\}$, which will be canceled in a later step. If such $D_v \succ 0$ exists, agent $v$ sends a probe to $w$ containing $[\mathbb{D}, D_v]$ and the information in Algorithm~\ref{alg:edge_addition}.
\end{itemize}

When the probe reaches $j$, unlike in Algorithm~\ref{alg:edge_addition}, the perturbation $\Delta$ is constructed directly from these $D_i$ using \eqref{eq:cycleL}. Here $E_{\text{add}} = \emptyset$, and $|E_{\text{add}}|=0$ is minimal by Theorem~\ref{thm:main-edge-removal}.

As a direct extension of Theorem~\ref{thm:main-edge-removal}, Theorem~\ref{thm:edge_addition}, and Theorem~\ref{thm:agent_joining}, the constructions for Cases (i)–(iii) provide explicit updates $\Delta$ such that, under the stated conditions, $L^+$ has a formation spectrum and $(G^+,\tilde{p}^+,R)$ satisfies Assumption~\ref{ass:graph}. In each case, the added edge set $E_{\text{add}}$ (explicitly given above) is of minimal size as characterized by Theorem~\ref{thm:main-edge-removal}.

\begin{remark}
From a graph-theoretic perspective, Theorem~\ref{thm:main-edge-removal} addresses a specific augmentation problem: given a $2$-vertex-connected graph $G$ and an edge $e$ to be removed, how to find a smallest set of \emph{new edges} (distinct from $e$) that restores $2$-vertex-connectivity to $G-e$. This contrasts with classical graph augmentation problems \cite{Eswaran1976,Rosenthal1977}, which typically consider adding edges to an arbitrary graph to achieve biconnectivity without the constraint of excluding a specific edge. Our contribution extends beyond this combinatorial characterization: we translate each case into a constructive Laplacian update protocol that maintains the desired spectral properties required for formation control.
\end{remark}

\begin{figure}[!t]
\centering
\scriptsize
\begin{subfigure}[b]{0.23\textwidth}
\centering
\begin{tikzpicture}[scale=0.45, every node/.style={scale=0.7}]
    \coordinate (v1) at (-3,3);
    \coordinate (v2) at (3,2);
    \coordinate (v3) at (2,0);
    
    \foreach \point/\name/\pos in {v1/1/above, v2/2/right, v3/3/right} {
        \filldraw[black] (\point) circle (4pt) node[\pos, yshift=0pt, xshift=4pt] {\name};
    }
    
    \draw[thick] (v1) -- node[midway, above] {$[\begin{smallmatrix}
        5 & 0 \\
        0 & -6
    \end{smallmatrix}]$} (v2);
    \draw[thick] (v2) -- node[midway, right] {$[\begin{smallmatrix}
        -30 & 0 \\
        0 & -3
    \end{smallmatrix}]$} (v3);
    \draw[thick] (v3) -- node[midway, below left] {$[\begin{smallmatrix}
        -6 & 0 \\
        0 & 2
    \end{smallmatrix}]$} (v1);
\end{tikzpicture}
\caption{\scriptsize Initial $\triangle_{1,2,3}$}
\label{fig:cycle_a}
\end{subfigure}
\hfill
\begin{subfigure}[b]{0.23\textwidth}
\centering
\begin{tikzpicture}[scale=0.45, every node/.style={scale=0.7}]
    \coordinate (v1) at (-3,3);
    \coordinate (v2) at (3,2);
    \coordinate (v3) at (2,0);
    \coordinate (v4) at (1,-1);
    
    \foreach \point/\name/\pos in {v1/1/above, v2/2/below, v3/3/below, v4/4/below} {
        \filldraw (\point) circle (4pt) node[\pos, yshift=0pt, xshift=4pt] {\name};
    }
    
    \draw[thick] (v1) -- node[midway, above] {$[\begin{smallmatrix}
        5 & 0 \\
        0 & -6
    \end{smallmatrix}]$} (v2);
    \draw[thick] (v2) -- node[midway, right] {$[\begin{smallmatrix}
        -30 & 0 \\
        0 & -3
    \end{smallmatrix}]$} (v3);
    \draw[thick, red, dashed] (v3) -- node[midway, below left] {} (v1);
    \draw[thick, red] (v3) -- node[midway, below, xshift=0.5cm] {$[\begin{smallmatrix}
        -30 & 0 \\
        0 & -6
    \end{smallmatrix}]$} (v4);
    \draw[thick, red] (v1) -- node[midway, left, xshift=-0.2cm] {$[\begin{smallmatrix}
        -7.5 & 0 \\
        0 & 1.5
    \end{smallmatrix}]$} (v4);
\end{tikzpicture}
\caption{\scriptsize Add vertex 4, remove $\{ 1,3 \}$}
\label{fig:cycle_b}
\end{subfigure}
\hfill
\begin{subfigure}[b]{0.23\textwidth}
\centering
\begin{tikzpicture}[scale=0.45, every node/.style={scale=0.7}]
    \coordinate (v1) at (-3,3);
    \coordinate (v2) at (3,2);
    \coordinate (v3) at (2,0);
    \coordinate (v4) at (1,-1);
    \coordinate (v5) at (0,-2);
    
    \foreach \point/\name/\pos in {v1/1/above, v2/2/below, v3/3/below, v4/4/below, v5/5/below} {
        \filldraw (\point) circle (4pt) node[\pos, yshift=0pt, xshift=4pt] {\name};
    }
    
    \draw[thick] (v1) -- node[midway, above] {$[\begin{smallmatrix}
        5 & 0 \\
        0 & -6
    \end{smallmatrix}]$} (v2);
    \draw[thick] (v2) -- node[midway, right, xshift=-0.2cm, yshift=-0.4cm] {$[\begin{smallmatrix}
        -30 & 0 \\
        0 & -3
    \end{smallmatrix}]$} (v3);
    \draw[thick] (v3) -- node[midway, right, xshift=-0.2cm, yshift=-0.4cm] {$[\begin{smallmatrix}
        -30 & 0 \\
        0 & -6
    \end{smallmatrix}]$} (v4);
    \draw[thick, red, dashed] (v1) -- node[midway, left] {} (v4);
    \draw[thick, red] (v4) -- node[midway, right, xshift=-0.2cm, yshift=-0.4cm] {$[\begin{smallmatrix}
        -30 & 0 \\
        0 & -6
    \end{smallmatrix}]$} (v5);
    \draw[thick, red] (v1) -- node[midway, left, xshift=0.1cm, yshift=-0.4cm] {$[\begin{smallmatrix}
        -10 & 0 \\
        0 & 1.2
    \end{smallmatrix}]$} (v5);
\end{tikzpicture}
\caption{\scriptsize Add vertex 5, remove $\{ 1,4 \}$}
\label{fig:cycle_c}
\end{subfigure}
\hfill
\begin{subfigure}[b]{0.23\textwidth}
\centering
\begin{tikzpicture}[scale=0.45, every node/.style={scale=0.7}]
    \coordinate (v1) at (-3,3);
    \coordinate (v2) at (3,2);
    \coordinate (v3) at (2,0);
    \coordinate (v4) at (1,-1);
    \coordinate (v5) at (0,-2);
    \coordinate (v6) at (-2,-3);
    
    \foreach \point/\name/\pos in {v1/1/above, v2/2/below, v3/3/below, v4/4/below, v5/5/below, v6/6/below} {
        \filldraw (\point) circle (4pt) node[\pos, yshift=0pt, xshift=4pt] {\name};
    }
    
    \draw[thick] (v1) -- node[midway, above] {$[\begin{smallmatrix}
        5 & 0 \\
        0 & -6
    \end{smallmatrix}]$} (v2);
    \draw[thick] (v2) -- node[midway, right, xshift=-0.2cm, yshift=-0.4cm] {$[\begin{smallmatrix}
        -30 & 0 \\
        0 & -3
    \end{smallmatrix}]$} (v3);
    \draw[thick] (v3) -- node[midway, right, xshift=-0.2cm, yshift=-0.4cm] {$[\begin{smallmatrix}
        -30 & 0 \\
        0 & -6
    \end{smallmatrix}]$} (v4);
    \draw[thick] (v4) -- node[midway, right, xshift=-0.2cm, yshift=-0.4cm] {$[\begin{smallmatrix}
        -30 & 0 \\
        0 & -6
    \end{smallmatrix}]$} (v5);
    \draw[thick, red, dashed] (v1) -- node[midway, left] {} (v5);
    \draw[thick, red] (v5) -- node[midway, right, xshift=-0.3cm, yshift=-0.4cm] {$[\begin{smallmatrix}
        -15 & 0 \\
        0 & -6
    \end{smallmatrix}]$} (v6);
    \draw[thick, red] (v1) -- node[midway, left] {$[\begin{smallmatrix}
        -30 & 0 \\
        0 & 1
    \end{smallmatrix}]$} (v6);
\end{tikzpicture}
\caption{\scriptsize Add vertex 6, remove $\{ 1,5 \}$}
\label{fig:cycle_d}
\end{subfigure}
\hfill
\begin{subfigure}[b]{0.23\textwidth}
\centering
\begin{tikzpicture}[scale=0.45, every node/.style={scale=0.7}]
    \coordinate (v1) at (-3,3);
    \coordinate (v2) at (3,2);
    \coordinate (v3) at (2,0);
    \coordinate (v4) at (1,-1);
    \coordinate (v5) at (0,-2);
    \coordinate (v6) at (-2,-3);
    \coordinate (v7) at (-1,-2.1);
    
    \foreach \point/\name/\pos in {v1/1/above, v2/2/below, v3/3/below, v4/4/below, v5/5/below, v6/6/below, v7/7/above} {
        \filldraw (\point) circle (4pt) node[\pos, yshift=0pt, xshift=4pt] {\name};
    }
    
    \draw[thick] (v1) -- node[midway, above] {$[\begin{smallmatrix}
        5 & 0 \\
        0 & -6
    \end{smallmatrix}]$} (v2);
    \draw[thick] (v2) -- node[midway, right, xshift=-0.2cm, yshift=-0.4cm] {$[\begin{smallmatrix}
        -30 & 0 \\
        0 & -3
    \end{smallmatrix}]$} (v3);
    \draw[thick] (v3) -- node[midway, right, xshift=-0.2cm, yshift=-0.4cm] {$[\begin{smallmatrix}
        -30 & 0 \\
        0 & -6
    \end{smallmatrix}]$} (v4);
    \draw[thick] (v4) -- node[midway, right, xshift=-0.2cm, yshift=-0.4cm] {$[\begin{smallmatrix}
        -30 & 0 \\
        0 & -6
    \end{smallmatrix}]$} (v5);
    \draw[thick] (v5) -- node[midway, right, xshift=-0.3cm, yshift=-0.4cm] {$[\begin{smallmatrix}
        -14 & 0 \\
        0 & -5.91
    \end{smallmatrix}]$} (v6);
    \draw[thick] (v1) -- node[midway, left] {$[\begin{smallmatrix}
        -30 & 0 \\
        0 & 1
    \end{smallmatrix}]$} (v6);
    \draw[thick, red] (v6) -- node[midway, left, xshift=0.3cm, yshift=0.4cm] {$[\begin{smallmatrix}
        -2 & 0 \\
        0 & -0.1
    \end{smallmatrix}]$} (v7);
    \draw[thick, red] (v5) -- node[midway, above, xshift=0cm, yshift=0.4cm] {$[\begin{smallmatrix}
        -2 & 0 \\
        0 & -0.9
    \end{smallmatrix}]$} (v7);
\end{tikzpicture}
\caption{\scriptsize Add vertex 7}
\label{fig:cycle_e}
\end{subfigure}
\hfill
\begin{subfigure}[b]{0.23\textwidth}
\centering
\begin{tikzpicture}[scale=0.45, every node/.style={scale=0.7}]
    \coordinate (v1) at (-3,3);
    \coordinate (v2) at (3,2);
    \coordinate (v3) at (2,0);
    \coordinate (v4) at (1,-1);
    \coordinate (v5) at (0,-2);
    \coordinate (v6) at (-2,-3);
    \coordinate (v7) at (-1,-2.1);
    
    \foreach \point/\name/\pos in {v1/1/above, v2/2/below, v3/3/below, v4/4/below, v5/5/below, v6/6/below, v7/7/above} {
        \filldraw (\point) circle (4pt) node[\pos, yshift=0pt, xshift=4pt] {\name};
    }
    
    \draw[thick] (v1) -- node[midway, above] {$[\begin{smallmatrix}
        5 & 0 \\
        0 & -6
    \end{smallmatrix}]$} (v2);
    \draw[thick] (v2) -- node[midway, right, xshift=-0.2cm, yshift=-0.4cm] {$[\begin{smallmatrix}
        -30 & 0 \\
        0 & -3
    \end{smallmatrix}]$} (v3);
    \draw[thick] (v3) -- node[midway, right, xshift=-0.2cm, yshift=-0.4cm] {$[\begin{smallmatrix}
        -30 & 0 \\
        0 & -6
    \end{smallmatrix}]$} (v4);
    \draw[thick] (v4) -- node[midway, right, xshift=-0.2cm, yshift=-0.4cm] {$[\begin{smallmatrix}
        -30 & 0 \\
        0 & -6
    \end{smallmatrix}]$} (v5);
    \draw[thick] (v1) -- node[midway, left] {$[\begin{smallmatrix}
        -30 & 0 \\
        0 & 1
    \end{smallmatrix}]$} (v6);
    \draw[thick] (v6) -- node[midway, left, xshift=0.4cm, yshift=0.4cm] {$[\begin{smallmatrix}
        -30 & 0 \\
        0 & -6.67
    \end{smallmatrix}]$} (v7);
    \draw[thick] (v5) -- node[midway, above, xshift=0cm, yshift=0.4cm] {$[\begin{smallmatrix}
        -30 & 0 \\
        0 & -60
    \end{smallmatrix}]$} (v7);
\end{tikzpicture}
\caption{\scriptsize Remove $\{ 5,6 \}$}
\label{fig:cycle_f}
\end{subfigure}

\caption{Edge weight evolution through agent joining and edge removal. Solid black edges: existing edges; dashed red edges: edges to be removed; red edges: newly added edges.}
\label{fig:cycle}
\end{figure}
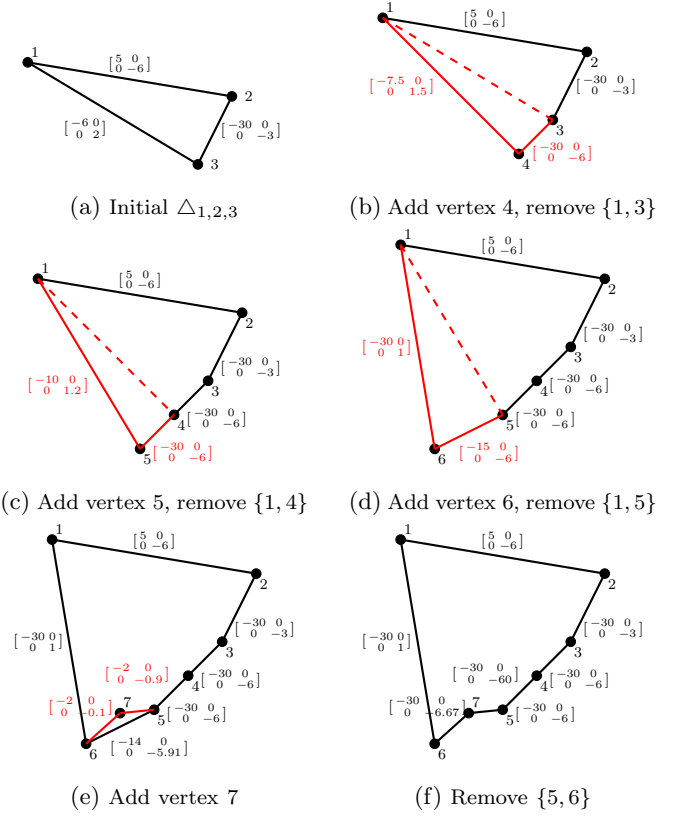

\section{Simulation}\label{sec:simulations}
This section presents numerical simulations to validate the theoretical results. Four experiments are conducted. The first constructs a matrix-valued Laplacian for a cycle graph, illustrating the iterative update mechanism described in \eqref{eq:L_update}. The other three couple the spectral property maintenance framework with the formation tracking control law from~\cite{He2026}, demonstrating formation maneuvers under dynamic topology changes, e.g., agent joining, agent leaving, and edge removal. (Experiments on edge addition and leader reassignment are omitted for brevity, as they follow the same pattern.)

For completeness, the control law is described as follows. Each leader $i \in V_l$ tracks a predefined trajectory via
\begin{equation}
u_i = -\alpha_1 \tanh(\alpha_2 (p_i - p^*_i)) + \dot{p}^*_i,
\label{equ_lu}
\end{equation}
where $p^*_i=S(R)\tilde{p}_i+\tau$, $\alpha_1,\alpha_2>0$ are control gains, and $\tanh(\cdot)$ denotes the hyperbolic tangent function. Each follower $i \in V_f$ maintains the formation using relative position measurements via
\begin{equation}
u_i = -\beta_1 \sum_{j \in N_i} L_{ij}(p_j-p_i) - \beta_2 \operatorname{sgn}\!\left(\sum_{j \in N_i} L_{ij}(p_j-p_i)\right),
\label{equ_pu}
\end{equation}
where $\beta_1,\beta_2>0$ are control gains, $\text{sgn}(\cdot)$ denotes the signum function, and $L_{ij}$ are the matrix-valued edge weights designed in previous sections. Agents 1 and 2 are designated as leaders, and the others are followers in all simulations.

\subsection{Laplacian Construction for a Cycle Graph} \label{sec:sim1}
This subsection constructs a matrix-valued Laplacian for a cycle graph using the iterative update mechanism \eqref{eq:L_update}. The resulting Laplacian will be used as the initial Laplacian in the subsequent formation maneuver control experiments.

We construct a cycle formation in $\mathbb{R}^2$ with $R = I_2$, starting from the triangle $\triangle_{1,2,3}$ with nominal positions $\tilde{p}_1 = [-3, 3]^\top$, $\tilde{p}_2 = [3, 2]^\top$, $\tilde{p}_3 = [2, 0]^\top$. For the initial triplet $(i,j,v)=(1,2,3)$, we form the constraint matrix $M_{123}$ as in \eqref{eq:iju_matrix} and set $D=I_2$ in \eqref{eq:ijv_block}. The resulting Laplacian $L^{123}=M_{123}^\top D M_{123}$ yields the edge weight matrices (the off-diagonal blocks of $L^{123}$) displayed on each edge in Fig.~\ref{fig:cycle_a}.

To add vertex $4$ with $\tilde{p}_4 = [1,-1]^\top$ while removing edge $\{1,3\}$, we combine the agent joining and edge removal operations. Following \eqref{eq:LVA}, we first pad $L^{123}$ to accommodate vertex $4$: $\mathcal{E}(L^{123}) = [\begin{smallmatrix} L^{123} & 0 \\ 0 & 0 \end{smallmatrix}]$. For the triplet $(i,j,v)=(1,3,4)$, we construct $L^{134}(D)$ as in \eqref{eq:ijv_block} and choose $D = \operatorname{diag}(1.5,\,0.5)$ to satisfy $W_{34}^\top D W_{41} + L^{123}_{13} = 0$, which removes edge $\{1,3\}$.  We then initialize a $4d\times 4d$ zero matrix $L^{134}_{\text{pad}}(D)$ and embed $L^{134}(D)$ into it by setting $(a_1,a_2,a_3)=(1,3,4)$ and $\bigl(L^{134}_{\text{pad}}(D)\bigr)_{a_ba_d} = \bigl(L^{134}(D)\bigr)_{bd}$ for $b,d\in\{1,2,3\}$. The updated Laplacian $L^+ = \mathcal{E}(L_{123}) + L^{134}_{\text{pad}}(D)$ corresponds to the graph shown in Fig.~\ref{fig:cycle_b}, in which edge $\{1,3\}$ is eliminated and new edges $\{1,4\},\{3,4\}$ appear. Vertices 5 and 6 are subsequently added using the same strategy: vertex 5 at $\tilde{p}_5 = [0, -2]^\top$ (connected to 1 and 4, deleting $\{ 1, 4 \}$) with $D = \diag(2.5, 0.3)$, and vertex 6 at $\tilde{p}_6 = [-2, -3]^\top$ (connected to 1 and 5, deleting $\{ 1, 5 \}$) with $D = \diag(5, 0.2)$. 

Next, vertex 7 is added at $\tilde{p}_7 = [-1, -2.1]^\top$ using the triplet $(i,j,v) = (5,6,7)$ with $D = \diag(1, 1)$, following the agent joining operation, as illustrated in Fig.~\ref{fig:cycle_e}. Finally, edge $\{5,6\}$ is deleted by applying the edge removal operation with the same triplet $(5,6,7)$: we require $L_{56} + W_{67}^\top D W_{75} = 0$, giving $D = \operatorname{diag}(14, 65.6667)$. The final cycle is shown in Fig.~\ref{fig:cycle_f}.

\begin{figure}[!tbp]
    \centering

    \begin{subfigure}[b]{0.45\textwidth} 
        \centering
        \includegraphics[width=\linewidth]{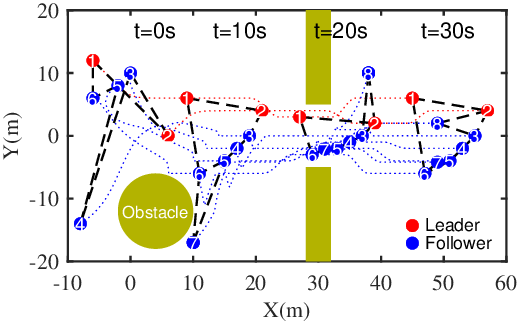}
        \caption{Agent trajectories.}
        \label{fig_te_va}
    \end{subfigure}
    \hfill
    \begin{subfigure}[b]{0.45\textwidth}
        \centering
        \includegraphics[width=\linewidth]{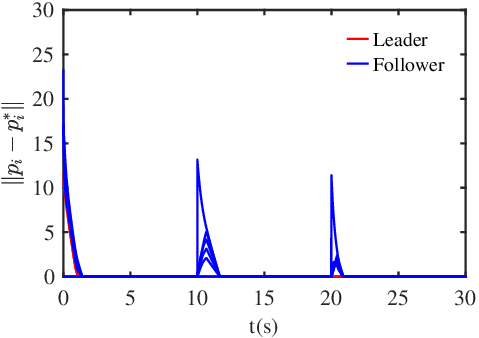}
        \caption{Tracking errors.}
        \label{fig_t_va}
    \end{subfigure}
    \caption{Formation maneuver control under agent joining (agent $7$ joins at $t=10$s, agent $8$ joins at $t=20$s).} 
    \label{fig_cycle_va}
\end{figure}

\subsection{Formation Maneuver Control with Agent Joining}
This simulation validates the proposed agent joining protocol (Algorithm~\ref{alg:agent_joining}) in a dynamic formation maneuver scenario. The initial six-agent formation is a cycle graph whose vertex set, edge set, and edge weight matrices are shown in Fig.~\ref{fig:cycle_d}. The Laplacian $L$ of this formation satisfies the formation spectrum (Definition~\ref{def:formation_spectrum}), as constructed in Section~\ref{sec:sim1}. The simulation timeline includes formation establishment (0--10 s), the integration of agent 7 at $10$ s, a non-uniform scaling maneuver (15--25 s), and the integration of agent 8 at $20$ s.

As shown in Fig.~\ref{fig_cycle_va}(a), the initial six-agent formation achieves and maintains the desired shape during the establishment phase, using the nominal formation and Laplacian from Subsection~\ref{sec:sim1}. At $t = 10$ s, agent 7 joins at $\tilde{p}_7 = [-1, -2.1]^\top$ by executing Algorithm~\ref{alg:agent_joining} with the triplet $(5,6,7)$ and $D=I_2$. This triggers a one-time distributed update of the Laplacian to $L^+$, which again satisfies the formation spectrum. From $t=15$ s to $t=25$ s, the formation performs a non-uniform scaling maneuver, compressing by 50\% along the y-axis to navigate a narrow passage. This specific maneuver is necessary since uniform scaling or rigid-body motion would result in inter-agent collisions or an inability to pass through the constriction. At $t = 20$ s, during the active scaling, agent 8 joins at $\tilde{p}_8 = [-1, 1]^\top$ by executing Algorithm~\ref{alg:agent_joining} with the triplet $(2,3,8)$ and $D = I_2$, testing the framework's ability to handle simultaneous shape deformation and network expansion.

Fig. \ref{fig_cycle_va} (b) shows that both leader tracking error and follower tracking error converge to zero, with only a minor transient disturbance during the agent joining event at $10$ s and $20$ s.

\begin{figure}[tbp]
    \centering

    \begin{subfigure}[b]{0.4\textwidth} 
        \centering
        \includegraphics[width=\linewidth]{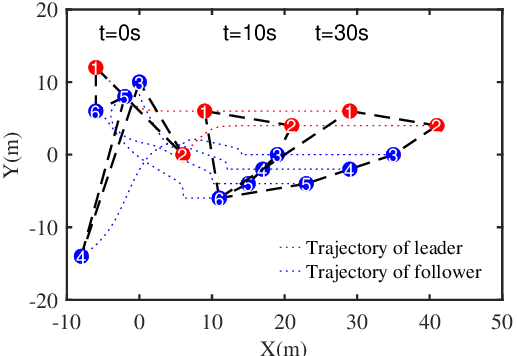}
        \caption{Agent trajectories.}
        \label{fig_te_vf}
    \end{subfigure}
    \hfill
    \begin{subfigure}[b]{0.4\textwidth}
        \centering
        \includegraphics[width=\linewidth]{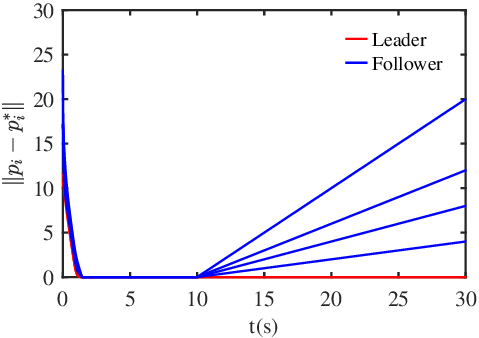}
        \caption{Tracking errors.}
        \label{fig_t_vf}
    \end{subfigure}
    \caption{Formation maneuver control without agent leaving protocol (agent $6$ halts at $t=10$ s).}
    \label{fig_cycle_vf}
\end{figure}

\begin{figure}[tbp]
    \centering

    \begin{subfigure}[b]{0.4\textwidth} 
        \centering
        \includegraphics[width=\linewidth]{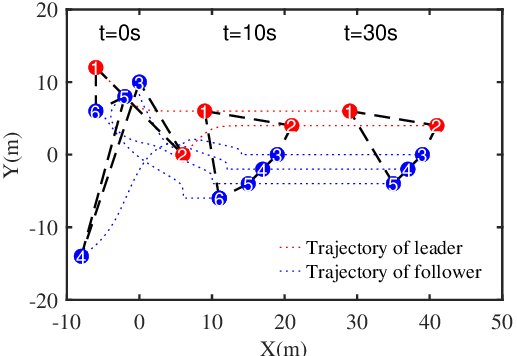}
        \caption{Agent trajectories.}
        \label{fig_te_vd}
    \end{subfigure}
    \hfill
    \begin{subfigure}[b]{0.4\textwidth}
        \centering
        \includegraphics[width=\linewidth]{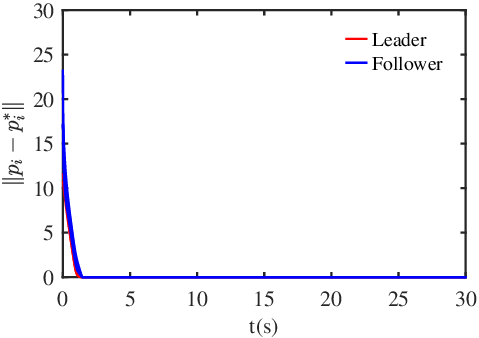}
        \caption{Tracking errors.}
        \label{fig_t_vd}
    \end{subfigure}
    \caption{Formation maneuver control with proposed agent leaving protocol (agent $6$ removed).}
    \label{fig_cycle_vd}
\end{figure}

\begin{figure}[tbp]
    \centering

    \begin{subfigure}[b]{0.4\textwidth} 
        \centering
        \includegraphics[width=\linewidth]{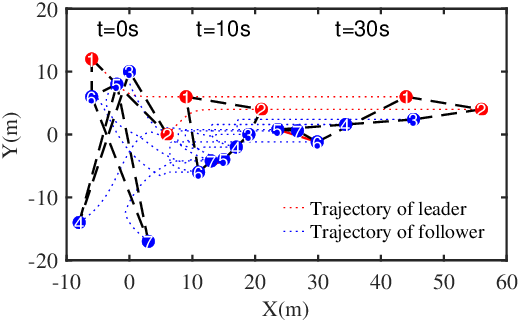}
        \caption{Agent trajectories.}
        \label{fig_te_ef}
    \end{subfigure}
    \hfill
    \begin{subfigure}[b]{0.4\textwidth}
        \centering
        \includegraphics[width=\linewidth]{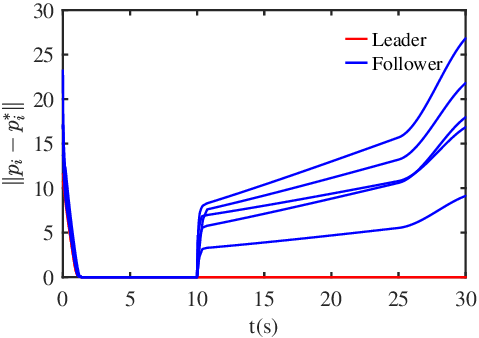}
        \caption{Tracking errors.}
        \label{fig_t_ef}
    \end{subfigure}
    \caption{Formation maneuver control without edge removal protocol (edge $\{5,6\}$ disconnected at $t=10$ s).}
    \label{fig_cycle_ef}
\end{figure}

\begin{figure}[tbp]
    \centering

    \begin{subfigure}[b]{0.4\textwidth} 
        \centering
        \includegraphics[width=\linewidth]{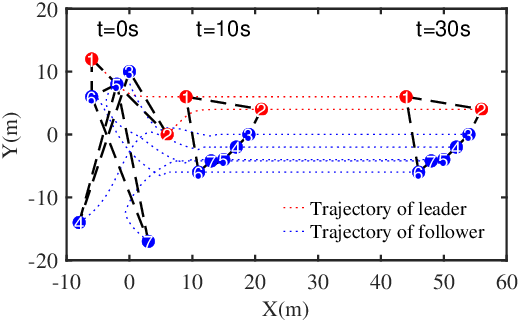}
        \caption{Agent trajectories.}
        \label{fig_te_ed}
    \end{subfigure}
    \hfill
    \begin{subfigure}[b]{0.4\textwidth}
        \centering
        \includegraphics[width=\linewidth]{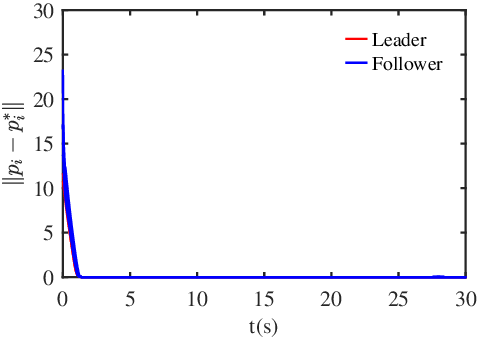}
        \caption{Tracking errors.}
        \label{fig_t_ed}
    \end{subfigure}
    \caption{Formation maneuver control with proposed edge removal protocol (edge $\{5,6\}$ removed at $t=10$ s).}
    \label{fig_cycle_ed}
\end{figure}

\subsection{Formation Maneuver Control with Agent Leaving}
This simulation validates the agent leaving protocol proposed in Section~\ref{sol:agent_leaving}. The initial formation is a graph whose vertex set, edge set, and edge weight matrices are shown in Fig.~\ref{fig:cycle_d}. The Laplacian $L$ of this formation satisfies the formation spectrum (Definition~\ref{def:formation_spectrum}), as constructed in Section~\ref{sec:sim1}. 

The simulation consists of two phases: formation establishment (0–10 s) and agent leaving (10–30 s). At $t = 10$ s, agent $6$ stops moving due to failure. (i.e., $p_6$ no longer changes), but the formation control law continues to run. Figures~\ref{fig_cycle_vf} and~\ref{fig_cycle_vd} compare the system behavior with and without the proposed agent leaving strategy. In Fig.~\ref{fig_cycle_vf}, the Laplacian is left unchanged; we simply set the velocity of agent $6$ to zero to simulate a failure, causing significant deformation and diverging tracking errors. In Fig.~\ref{fig_cycle_vd}, we apply the agent leaving protocol: $\bar\Delta = 0$ and $\Delta = -L_u^\top L_{uu}^{-1} L_u$ with $u = 6$, yielding $L^+ = \mathcal{E}(L + \Delta)$. The updated Laplacian corresponds to the five-agent graph in Fig.~\ref{fig:cycle_c}. Using $L^+$ in the control law, the remaining agents rapidly reconfigure into a consistent formation while keeping all tracking errors bounded.

\subsection{Formation Maneuver Control with Edge Removal}
This simulation validates the edge removal protocol for Case (i) of Subsection~\ref{subsec:edge_removal}, where $H=G-e$ remains $2$-vertex-connected and no compensatory edges are needed ($E_{\text{add}} = \emptyset$). The initial formation is a graph whose vertex set, edge set, and edge weight matrices are shown in Fig.~\ref{fig:cycle_e}. The Laplacian $L$ of this formation satisfies the formation spectrum (Definition~\ref{def:formation_spectrum}), as constructed in Section~\ref{sec:sim1}.

The simulation consists of two phases: formation establishment (0–10 s) and edge removal (10–30 s). At $t=10$ s, the relative displacement measurement of edge $\{5,6\}$ is lost. Figures~\ref{fig_cycle_ef} and~\ref{fig_cycle_ed} compare the system behavior with and without the proposed edge removal strategy. In Fig.~\ref{fig_cycle_ef}, the Laplacian is left unchanged; the missing measurement is set to zero, leading to visible formation distortion and unbounded tracking errors. In Fig.~\ref{fig_cycle_ed}, the proposed edge removal protocol is executed: using triplet $(5,6,7)$, we compute $D = \operatorname{diag}(14,\;65.6667)$ to satisfy $W_{67}^\top D W_{75} = -L_{56}$, construct $\Delta = L^{567}_{\text{pad}}(D)$ (as in \eqref{eq_pad}), and update the Laplacian to $L^+ = L + \Delta$, corresponding to the cycle graph in Fig.~\ref{fig:cycle_f}. The updated Laplacian maintains the formation spectrum; the formation rapidly converges again and tracking errors remain bounded. 

\section{Conclusion}\label{sec:conclusion}

We have developed a distributed framework for non-uniform scaling formation maneuver control of open multi-agent systems. It handles four elementary topology operations—agent joining, agent leaving, edge addition, and edge removal—solving the spectral shaping problem of the matrix-valued Laplacian under topology changes. Rigorous theoretical analysis and simulation validation were provided.

Several directions merit future investigation. Extending the proposed framework to directed graphs would broaden its applicability. Adaptive leader–follower reassignment and maneuver parameter generation based on environmental feedback present promising avenues. Graph augmentation under more complex constraints also warrants further study.

\appendix
\section{Proof of Theorem \ref{thm:agent_joining}} \label{sec:proofTheagent_joining}
The proof requires a lemma.
\begin{lemma}[Schur complement \cite{horn2013matrix,zhang2005}]
\label{lem:inverse_schur}
For any symmetric block matrix of the form $M=[\begin{smallmatrix}A & B\\ B^\top & C \end{smallmatrix}]$ where $C$ is invertible, let $S=A-BC^{-1}B^\top$ denote the Schur complement of $C$ in $M$. Then the following statements hold:
\begin{enumerate}
    \item[(1)] $M$ is invertible if and only if $S$ is invertible. In that case,
    \[
    M^{-1}=
    \begin{bmatrix}
    S^{-1} & -S^{-1}BC^{-1}\\
    -C^{-1}B^\top S^{-1} &
    C^{-1}+C^{-1}B^\top S^{-1}BC^{-1}
    \end{bmatrix}.
    \]

    \item[(2)] $M\succ0$ if and only if $C\succ0$ and $S\succ0$.

    \item[(3)] If $C \succ 0$, then $M \succeq 0$ if and only if $S \succeq 0$.
\end{enumerate}
\end{lemma}

We now prove Theorem~\ref{thm:agent_joining}.
\begin{pf}
We first prove that $L^+$ has a formation spectrum by verifying the three conditions in Definition \ref{def:formation_spectrum}.

\textbf{(1) Proof of $L^+ \succeq 0$ (condition (i) of Definition \ref{def:formation_spectrum}).}
Since $L \succeq 0$ and $L^{ijv} \succeq 0$, we have $L^+ \succeq 0$. 

\textbf{(2) Proof of $Q^\top L^+ Q \succ 0$ for every $Q \in \mathcal{Q}$ (condition (iii) of Definition \ref{def:formation_spectrum}).} For the specific leader set ${r_1,r_2}$, this reduces to proving $L^+_{ff} \succ 0$. To prove this, we consider two cases: $v$ as a follower and $v$ as a leader. The proof for other leader sets is analogous.

\textbf{Case A: $v$ is added as a follower.} For the detailed proof, see \cite{He2026} (the preliminary version of this work).

\textbf{Case B: $v$ is added as a leader.}  
To maintain exactly two leaders, we replace one of the original leaders with $v$. Without loss of generality, we retain $r_1$ as a leader and demote $r_2$ to a follower. The new leader set is $\{r_1,v\}$, and the new follower set is $V_f \cup \{r_2\}$. We now consider all possible roles of the selected neighbours $i,j$ and show that the new follower block $L^+_{ff}$ is positive definite.

Partition the new follower set into $P_1 = V_f \setminus \{i,j\}$ and $P_2 = (V_f \cup \{r_2\}) \setminus P_1$. Write
\begin{equation}
L^+_{ff} = \begin{bmatrix}
A & B \\
B^\top & C + \Delta_C
\end{bmatrix},
\end{equation}
where $A = L_{P_1P_1} \succ 0$, $C= L_{P_2P_2}$, and $\Delta_C = \bar{\Delta}_{P_2P_2} \succeq 0$. By the Schur complement lemma, $L^+_{ff} \succ 0$ if and only if $S + \Delta_C \succ 0$ with $S = C - B^\top A^{-1} B \succeq 0$.

Suppose there exists a nonzero $y \in \ker(S) \cap \ker(\Delta_C)$. Then from $y \in \ker(S)$, there exists $x$ such that $[\begin{smallmatrix} A & B \\ B^\top & C \end{smallmatrix}] [\begin{smallmatrix} x \\ y \end{smallmatrix}] = 0$. Define $w$ on the original vertex set $V$ by setting $w|_{P_1}=x$, $w|_{P_2}=y$, and $w_{r_1}=0$. Because \(w^\top L w = [\begin{smallmatrix} x \, y \end{smallmatrix}] [\begin{smallmatrix} A & B \\ B^\top & C \end{smallmatrix}] [\begin{smallmatrix} x \\ y \end{smallmatrix}] = 0\) and \(L \succeq 0\), we have \(L w = 0\). Hence $w$ lies in $\varPi(\tilde{p},R)$, i.e., there exist $S(R)$ and $\tau$ such that $w_k = \tau + S(R)\tilde{p}_k$ for all $k$. With $w_{r_1}=0$, we get $\tau = -S(R)\tilde{p}_{r_1}$, and thus for any $k \in P_2$, 
\begin{equation}\label{eq_leaderv}
    w_k = S(R)(\tilde{p}_k - \tilde{p}_{r_1}).
\end{equation}

Now we examine each subcase, using the condition \(y \in \ker(\Delta_C)\) to derive a contradiction.

\textbf{Subcase B1: $i=r_1$, $j=r_2$.}  
Here $P_2 = \{r_2\}$ and $\Delta_C = W_{v r_1}^\top D W_{v r_1} \succ 0$, so $\ker(\Delta_C)=\{0\}$. Hence no nonzero $y$ exists.

\textbf{Subcase B2: $i=r_1$, $j\in V_f$.}  
Here $P_2 = \{j, r_2\}$ and $\Delta_C = \operatorname{diag}(W_{v r_1}^\top D W_{v r_1}, 0)$ in the ordering $(j, r_2)$. If $y \in \ker(\Delta_C)$, then $y_j=0$. From \eqref{eq_leaderv}, $y_j = S(R)(\tilde{p}_j - \tilde{p}_{r_1})$, so $\tilde{p}_j = \tilde{p}_{r_1}$. This contradicts the fact that $j$ and $r_1$ are distinct agents, which is guaranteed by Assumption~\ref{ass:graph}.

\textbf{Subcase B3: $i\in V_f$, $j=r_2$.}  
Here $P_2 = \{i, r_2\}$ and 
\begin{equation}
\Delta_C = \begin{bmatrix} W_{r_2 v}^\top D W_{r_2 v} & W_{r_2 v}^\top D W_{v i} \\ W_{v i}^\top D W_{r_2 v} & W_{v i}^\top D W_{v i} \end{bmatrix}.
\end{equation}
Its nullspace is $\ker(\Delta_C) = \{y=[y_i, y_{r_2}]^{\top} : W_{r_2 v} y_i + W_{v i} y_{r_2} = 0\}$. If \(y \in \ker(\Delta_C)\), then using \eqref{eq_leaderv} for \(i\) and \(r_2\) and substituting into the nullspace condition yields after simplification
\begin{equation}
(\tilde{p}_{i,R}^\alpha - \tilde{p}_{r_2,R}^\alpha)(\tilde{p}_{r_1,R}^\alpha - \tilde{p}_{v,R}^\alpha) = 0 \quad \forall \alpha \in \{1,\dots,d\}.
\end{equation}
By Assumption~\ref{ass:graph}, $\tilde{p}_{i r_2,R}$ has all components nonzero, so $\tilde{p}_{i,R}^\alpha \neq \tilde{p}_{r_2,R}^\alpha$. Hence $\tilde{p}_{r_1,R}^\alpha = \tilde{p}_{v,R}^\alpha$ for all $\alpha$, i.e., $\tilde{p}_{r_1} = \tilde{p}_v$, contradicting that $v$ is a new agent with distinct nominal position (can be ensured by choice). Thus no such $y$ exists.

\textbf{Subcase B4: $i,j\in V_f$.}  
Here $P_2 = \{i, j, r_2\}$ and $\Delta_C$ affects the $(i,j,r_2)$ block. Its nullspace is $\ker(\Delta_C) = \{y=[y_i,y_j,y_{r_2}]^{\top}: W_{jv} y_i + W_{vi} y_j = 0\}$. If \(y \in \ker(\Delta_C)\), then using \eqref{eq_leaderv} for $i,j,r_2$ and substituting into the nullspace condition yields after simplification
\begin{equation}
(\tilde{p}_{i,R}^\alpha - \tilde{p}_{j,R}^\alpha)(\tilde{p}_{r_1,R}^\alpha - \tilde{p}_{v,R}^\alpha) = 0 \quad \forall \alpha \in \{1,\dots,d\}.
\end{equation}
By Assumption~\ref{ass:graph}, $\tilde{p}_{ji,R}$ has all components nonzero, so $\tilde{p}_{i,R}^\alpha \neq \tilde{p}_{j,R}^\alpha$. Hence $\tilde{p}_{r_1,R}^\alpha = \tilde{p}_{v,R}^\alpha$ for all $\alpha$, i.e., $\tilde{p}_{r_1} = \tilde{p}_v$, again a contradiction.

In all subcases, no nonzero $y \in \ker(S) \cap \ker(\Delta_C)$ exists. Therefore $S + \Delta_C \succ 0$, and consequently $L^+_{ff} \succ 0$. 

\textbf{(3) Proof of $\ker(L^+) = \varPi(\tilde{p}^+,R)$ (condition (ii) of Definition \ref{def:formation_spectrum}).}
First, by construction and Lemma \ref{lem:2}, any $p^+ \in \varPi(\tilde{p}^+,R)$ satisfies:
\begin{equation}
\mathcal{E}(L) p^+ = 0 \quad \text{and} \quad L^{ijv}_{\text{pad}} p^+ = 0,
\end{equation}
implying that $L^+ p^+ = 0$. This shows $\varPi(\tilde{p}^+,R) \subseteq \ker(L^+)$. For the reverse inclusion, note that $L^+_{ff}\succ 0$ (proved above) and $L^+_{ff}\in\mathbb{R}^{d(n-1)\times d(n-1)}$ imply $\operatorname{rank}(L^+)\ge d(n-1)$. Since $\tilde{p}^+$ contains $\tilde{p}$ as a subvector, the projection $\pi:\mathbb{R}^{d(n+1)}\to\mathbb{R}^{dn}$ onto the first $n$ agents satisfies $\pi(\varPi(\tilde{p}^+,R))=\varPi(\tilde{p},R)$. As $\dim(\varPi(\tilde{p},R))=2d$ and $\pi$ is surjective onto it, we have $\dim(\varPi(\tilde{p}^+,R))\ge 2d$. Every element of $\varPi(\tilde{p}^+,R)$ is parameterized by $(s,\tau)\in\mathbb{R}^{2d}$, so $\dim(\varPi(\tilde{p}^+,R))\le 2d$. Hence $\dim(\varPi(\tilde{p}^+,R))=2d$. Consequently, $\operatorname{rank}(L^+)\le d(n+1)-2d = d(n-1)$, which together with the lower bound gives $\operatorname{rank}(L^+)=d(n-1)$ and $\dim(\ker(L^+))=2d$. The inclusion $\varPi(\tilde{p}^+,R) \subseteq \ker(L^+)$ together with $\dim(\varPi(\tilde{p}^+,R)) = \dim(\ker(L^+)) = 2d$ implies $\ker(L^+) = \varPi(\tilde{p}^+,R)$.

\textbf{(4) Proof of minimality of $|E_{\text{add}}|$ and $|E_{\text{mod}}|$.}
We choose $D\succ0$ such that $L_{ij} + W_{jv}^\top D W_{vi} \neq 0$. By Assumption~\ref{ass:graph}, $W_{jv}$ and $W_{vi}$ are invertible. Consequently, the construction of $L^+$ in \eqref{eq:LVA} adds precisely two new edges $\{i,v\}$ and $\{j,v\}$ (whose weights are non-zero), while preserving the existing edge $\{i,j\}$ and leaving all other edges unchanged. Hence $E^+ = E \cup E_{\text{add}}$, and $E_{\text{mod}} = \{\{i,j\}\}$. We now verify that $|E_{\text{add}}|$ is minimal while ensuring $G^+$ remains 2-vertex-connected. Note that $L^+$ is symmetric, so $G^+$ is undirected. Adding only one edge $\{i,v\}$, would violate 2-vertex-connectivity because removing $i$ would isolate $v$. Hence at least two edges are necessary. To verify that $G^+$ with $E_{\text{add}}$ indeed preserves 2-vertex-connectedness, remove any vertex $x\in V^+$. If $x=v$, the remaining graph is $G$, which is connected. If $x\in V$ and $x\neq i,j$, then $G-x$ contains a connected subgraph with $i,j$; since $v$ connects to both, $G^+-x$ remains connected. If $x=i$ (or symmetrically $x=j$), $G-i$ is connected and contains $j$; with $v$ connected to $j$, $G^+-i$ stays connected. Thus $G^+$ is 2-vertex-connected.

To establish minimality of $|E_{\text{mod}}|$, suppose $|E_{\text{mod}}| = 0$. Then $L^+ = \mathcal{E}(L) + \bar{\Delta}$ with $\bar{\Delta} \neq 0$ only affects rows and columns of $v$. Since $L^+$ and $\mathcal{E}(L)$ are symmetric, $\bar{\Delta}$ is symmetric. For any $p\in\varPi(\tilde{p}^+,R)$, $L^+p = \mathcal{E}(L)p = 0$, hence $\bar{\Delta}p=0$. Choose $p = \mathbf{1}_{n+1}\otimes\tau$ with $\tau\neq0$. Then $\bar{\Delta}_{iv}\tau = 0$ for all $\tau$ and any $i \in V$, forcing $\bar{\Delta}_{iv}=0$. By symmetry, $\bar{\Delta}_{vi}=0$ for any $i \in V$, thus $\bar{\Delta}=0$, contradicting $\bar{\Delta}\neq0$. Therefore $|E_{\text{mod}}|\ge1$. Our construction achieves $|E_{\text{mod}}|=1$, which is minimal.
\end{pf}

\bibliographystyle{plain}        
\bibliography{references}           


\end{document}